\documentclass[sigconf, nonacm]{acmart}
\usepackage{newfile}
\usepackage{multirow}
\usepackage{enumerate}
\usepackage{caption} 
\usepackage{titlesec}

\newcommand{\NP}{\textsf{NP}}
\newcommand{\PSPACE}{\textsf{PSPACE}}
\newcommand{\NEXP}{\textsf{NEXP}}

\AtBeginDocument{%
	}

\usepackage{listings}
\lstdefinelanguage{pseudo}{morekeywords={init,with,or,if,then,else,fi,and,not,while,do,od,distinct,
    case, goto,local,algorithm, function, for, each, times, from, to,
    variables, procedure, recursive, return},
  morecomment=[l]{//}, morecomment=[s]{/*}{*/},
  mathescape=true,escapechar={@},
  basicstyle=\sffamily\small,
  commentstyle=\itshape\rmfamily\small,
  keywordstyle=\sffamily\bfseries\small
}
\usepackage{mathtools}
\usepackage{amsfonts,amsmath}

\usepackage{etoolbox}  
\usepackage{enumitem} 

\usepackage{stmaryrd}
\usepackage{bm}

\usepackage{hyperref}
\usepackage{environ}
\usepackage{csvsimple}
\usepackage{longtable} 
\usepackage[normalem]{ulem} 

\usepackage{fourier-orns} 
\usepackage{wasysym} 
\usepackage{marvosym} 
\usepackage{bbding} 

\usepackage{etoolbox}  
\usepackage{enumitem} 

\usepackage{enumerate}
\usepackage{mathtools}
\usepackage{appendix}
\usepackage{comment}
\usepackage{tikz}
\usepackage{array}
\usetikzlibrary{calc}
\usepackage{multicol}

\usetikzlibrary{arrows,shapes,snakes,automata,backgrounds,petri,positioning}
\usetikzlibrary{fit,backgrounds} 
\definecolor{processblue}{cmyk}{0.96,0,0,0}

\makeatletter
\newcommand{\dashact}[2][]{\ext@arrow 0359\rightarrowfill@@{#1}{#2}}
\def\rightarrowfill@@{\arrowfill@@\relax\relbar\dashrightarrow}
\def\arrowfill@@#1#2#3#4{%
	$\m@th\thickmuskip0mu\medmuskip\thickmuskip\thinmuskip\thickmuskip
	\relax#4#1
	\xleaders\hbox{$#4#2$}\hfill
	#3$%
}
\makeatother

\newcommand{\nn}{\mathbb{N}}

\newcommand{\zn}{\mathbb{Z}}
\newcommand{\qn}{\mathbb{Q}}
\newcommand{\qnz}{\mathbb{Q}_{+}}

\newcommand{\be}{\begin{enumerate}}
\newcommand{\ee}{\end{enumerate}}
\newcommand{\bc}{\begin{center}}
\newcommand{\ec}{\end{center}}

\newcommand{\bi}{\begin{itemize}}
\newcommand{\ei}{\end{itemize}}

\newcommand{\act}{\xrightarrow}

\newcommand{\gr}{\mathcal{G}}

\newcommand{\node}{\mathsf{n}}

\newcommand{\bA}{\mathbf{A}}
\newcommand{\bb}{\mathbf{b}}
\newcommand{\bI}{\mathbf{I}}
\newcommand{\bP}{\mathbf{P}}
\newcommand{\bB}{\mathbf{B}}
\newcommand{\restrict}{\mathord{\upharpoonright}}
\newcommand{\ext}[2]{{#1}\!\restrict_{#2}} 
\newcommand{\bzero}{\mathbf{0}}
\newcommand{\bunit}{\mathbf{e}}

\newcommand{\apply}[3]{\text{Apply}_{#1}(#2,#3)}

\usepackage[textsize=scriptsize]{todonotes}
\usepackage{hyperref}

\newcommand{\mach}{\mathcal{M}}

\newcommand{\bx}{\mathbf{x}}

\newcommand{\bz}{\mathbf{z}}
\newcommand{\bu}{\mathbf{u}}
\newcommand{\bv}{\mathbf{v}}
\newcommand{\bw}{\mathbf{w}}

\newcommand\slice[2]{#1{\raise-.5ex\hbox{\ensuremath|}}_{#2}}

\mathchardef\mhyphen="2D

\newcommand{\supp}[1]{{\llbracket#1\rrbracket}}  

\usepackage{thm-restate}

\def\cC{\mathcal{C}}

\makeatletter

\newcommand{\eqxrightarrow}[2]{%
  \mathop{%
    \vtop{%
      \m@th 
      \offinterlineskip 
      \ialign{%
        \hfil##\hfil\cr
        \rightarrowfill\cr
        \hphantom{$\scriptstyle\mskip8mu{#2}\mskip8mu$}\cr
        \vrule height0pt width 1.5em\cr
        $\scriptscriptstyle {#1}$\cr
      }%
    }%
  }\limits^{#2}%
}
\makeatother

\title{Decidability and Complexity of Decision Problems for Affine Continuous VASS}

\begin{abstract}
	Vector addition system with states (VASS) is a popular model for the verification
	of concurrent systems. VASS consists of finitely many control states and a set of counters
	which can be incremented and decremented, but not tested for zero. 
	VASS is a relatively well-studied model of computation and many results regarding 
	the decidability of decision problems for VASS are well-known.	
	Given that the complexity of solving almost all problems for VASS is very high, various
	tractable over-approximations of the reachability relation of VASS have been proposed in the literature.
	One such tractable over-approximation is the so-called \emph{continuous VASS}, in which counters are allowed
	to have non-negative rational values and whenever an update is performed, the update is first scaled
	by an arbitrary non-zero fraction. 
	
	In this paper, we consider \emph{affine} continuous VASS, which extend continuous VASS
	by allowing integer affine operations. Affine continuous VASS serve as an over-approximation
	to the model of affine VASS, in the same way that continuous VASS over-approximates the reachability relation of VASS. 
	We investigate the tractability of affine continuous VASS with respect to 
	the reachability, coverability and state-reachability
	problems for different classes of affine operations and we 
	prove an almost-complete classification of the decidability of these problems.
	Namely, except for the coverability problem for a single family of classes of affine operations, 
	we completely determine the decidability
	status of these problems for all classes. 
	Furthermore, except for this single family, we also complement
	all of our decidability results with tight complexity-theoretic upper and lower bounds. 
\end{abstract}

\author{A. R. Balasubramanian}
\email{bayikudi@mpi-sws.org}
\orcid{0000-0002-7258-5445}
\affiliation{%
	\institution{Max Planck Institute for Software Systems}
	\city{Kaiserslautern}
	\country{Germany}
}



\copyrightyear{2024}
\acmYear{2024}
\setcopyright{rightsretained}
\acmConference[LICS '24]{39th Annual ACM/IEEE Symposium on Logic in Computer Science}{July 8--11, 2024}{Tallinn, Estonia}
\acmBooktitle{39th Annual ACM/IEEE Symposium on Logic in Computer Science (LICS '24), July 8--11, 2024, Tallinn, Estonia}\acmDOI{10.1145/3661814.3662124}
\acmISBN{979-8-4007-0660-8/24/07}

\begin{document}

\maketitle

\section{Introduction}

Vector addition system with states (VASS) are one of the most popular and most studied infinite-state systems and have many applications in 
verifying and modeling concurrent and business processes~\cite{jacm/GermanS92,acta/EsparzaGLM17,tacas/DelzannoRB02,jcsc/Aalst98}. VASS consist of a finite set of states along
with a set of counters with values ranging over $\nn$. Hence, a configuration of a VASS is a pair consisting of the current state of the VASS along with the current values of its counters.
Transitions of the VASS allow it to move from one state to another while incrementing or decrementing the values of the counters, provided that the new values of the counters do not drop below zero. 
Alternatively, this can be summarized as having update instructions of the form $\bx \leftarrow \bx + \bb$
which are fireable as long as $\bx + \bb \ge \bzero$. 
One of the most central questions regarding VASS
is the \emph{reachability} problem: Given two configurations $C$ and $C'$ of a VASS, can $C$ reach $C'$?
After decades of being open, the complexity of the problem was recently shown 
to be \textsf{ACKERMANN}-complete~\cite{leroux2022reachability,czerwinskiReachabilityVectorAddition2022}. 

As mentioned before, the model of VASS has many applications. However, since basic verification questions
such as reachability have extremely high complexity, much of the research has been concentrated on finding
\emph{over-approximations} of the reachability relation which trade expressiveness for tractability.
Note that since reachability is usually used to check the existence of a path to an unsafe configuration,
if we prove that the unsafe state is not reachable in an over-approximate model, 
then the same holds for the original model as well.
Many heuristics and over-approximations for the reachability relation of a VASS are known~\cite{ABCs} and they have been applied in practice quite successfully~\cite{ABCs,tocl/BlondinFHH17}. 
Here we concentrate on one such over-approximation, namely the \emph{continuous semantics}.

In the continuous semantics, transitions are allowed to be fired \emph{fractionally}. 
More precisely, in order to fire a transition $t$, we are allowed to first pick a non-zero fraction $\alpha \in (0,1]$ and then execute the updates of $t$ by the fraction $\alpha$, i.e., if $t$ originally decremented (resp. incremented) a counter by $c$, we now decrement (resp. increment) that counter by $\alpha \cdot c$, as long as the new value does not drop below zero.
As a result of this ``fractional update'', counters are now relaxed to have values over the non-negative rationals $\qnz$. Note that the continuous semantics is an over-approximation of the usual semantics.
It turns out that the continuous semantics is quite well-behaved in terms of complexity: Reachability in continuous VASS is \NP-complete and is hence much lower than the complexity for the usual semantics.
Besides being a tractable over-approximation, it is also useful in practice~\cite{tocl/BlondinFHH17,tacas/BlondinHO21,cav/BlondinMO22}
and is also useful to derive sound and complete algorithms for various problems in parameterized
verification~\cite{lmcs/BalasubramanianER23,fsttcs/Balasubramanian20}.

Coming back to VASS, as mentioned in~\cite{Affine-VASS}, although they are useful for modeling many systems, the features of VASS might 
be too restricted in certain cases. For instance, VASS cannot model counter systems which
are allowed to reset their counters to 0, or transfer or copy the contents of one counter to another. It turns out that such extensions are quite desirable to
model certain applications~\cite{toplas/KW14,tacas/DelzannoRB02,cav/SilvermanK19}. Many of these extensions can be modeled
as specific \emph{classes of affine VASS}. Affine VASS are VASS extended with affine operations, i.e.,
transitions can have updates of the form $\bx \leftarrow \bA \bx + \bb$ for some matrix $\bA$
and some vector $\bb$ over the integers and such transitions are fireable as long as $\bA \bx + \bb \ge 0$.
However, this huge expressiveness comes at a cost; almost all interesting verification problems
are either undecidable or have very high complexity for affine VASS~\cite{Affine-VASS}.

As a result of this undecidability/high complexity barrier for affine VASS, it would then be 
natural to wonder if the tractability of over-approximations developed for VASS can also
be transferred to affine VASS. In particular, a natural question might be to ask
\begin{quote}
	Does the tractability of the continuous semantics hold in the affine case as well?
\end{quote}

\subsubsection*{Our contribution. } In this paper, we consider the above question and introduce the model
of \emph{affine continuous VASS}. An affine continuous VASS is just like a VASS, except that
transitions are now allowed to be fired fractionally, i.e., before firing a transition $t$ (whose instruction is given by $\bx \leftarrow \bA \bx + \bb$ where $\bA$ and $\bb$ are a matrix and a vector over the integers respectively), 
we pick a non-zero fraction $\alpha$ and update
the counters by the rule $\bx \leftarrow \bA \bx + \alpha \bb$ provided that $\bA \bx + \alpha \bb \ge \bzero$. To put it in words, to get the new counter values,
we multiply the old counter values with the matrix $\bA$ and then add the vector $\alpha \bb$,
provided that the final result we get does not make the value of any counter drop below zero.

At this point, the reader might wonder why the semantics for the continuous case was not defined so that the update looks like $\bx \leftarrow \alpha \bA \bx + \alpha \bb$. First, defining the semantics
this way does not allow us to model various important features such as resets or transfers. 
Second, we would ideally like the semantics that we define for affine continuous VASS to
extend the semantics for continuous VASS, i.e., any meaningful extension of continuous VASS
to the affine setting should allow for a continuous VASS to be interpreted as is in the extended semantics by thinking of each transition as being equipped with an identity matrix. If we define the semantics 
for the affine case by the rule $\bx \leftarrow \alpha \bA \bx + \alpha \bb$, then 
this feature is lost. 
For these two reasons, we choose the semantics of an affine continuous VASS
so that the \textbf{matrix is not multiplied with the chosen fraction} $\alpha$.

Having defined this model, we study three of the most fundamental verification problems: Reachability (can a given configuration reach another given configuration?), Coverability 
(given two configurations $(p,\bu)$ and $(q,\bv)$ can $(p,\bu)$ reach
some $(q,\bv')$ with $\bv' \ge \bv$?) and State-reachability (given a configuration $(p,\bu)$ and a state $q$, can $(p,\bu)$ reach some configuration whose state is $q$?). 
We present an almost-complete classification of the decidability status of these problems
with respect to the various \emph{classes of affine operations} that are allowed, i.e., with
respect to the various classes of matrices that are allowed. We now briefly summarize our main
contributions.

\paragraph*{Reachability. } As our first contribution, we completely resolve the decidability status of reachability. We show that reachability for a class of affine continuous VASS is decidable if
the class contains only permutation matrices and is undecidable for \emph{any other class}. 

\paragraph*{State-Reachability. } As our second contribution, we completely resolve
the decidability status of state-reachability. We show that state-reachability 
for a class of affine continuous VASS is decidable (and in \PSPACE) if it contains only non-negative matrices and is undecidable for any other class. This proves, for instance, that state-reachability is in
\PSPACE \ for continuous VASS with resets, transfers, copies or doubling operations. To the best of our knowledge, the \PSPACE\ upper bounds for these models were not known before.

\paragraph*{Coverability. } As our third contribution, except for a single family of classes of 
of affine matrices, we resolve the decidability status of coverability.
We show that coverability is decidable for classes containing only permutation matrices
and is undecidable for any class containing a matrix with a negative entry or a matrix with a row or a column that is entirely comprised of zeroes. Since continuous VASS with resets or transfers
belong to the last category of matrices, it follows that \emph{the coverability problem
for continuous VASS with resets or transfers is undecidable}. 
This result is surprising, because the coverability
problem for (usual) VASS with resets or transfers is decidable~\cite{wsts}. Hence, while
the continuous semantics drastically reduces the complexity for VASS, for resets and 
transfers it actually strictly increases in complexity and power.

As we show in the next section, the only remaining classes of matrices for which we were not able to prove a decidability or undecidability
result for coverability are non-negative classes without any zero-rows or columns which contain at least one matrix possessing a weighted edge or two overlapping edges.
We will explain these notions formally in Section~\ref{sec:prelims}, but we describe them informally here.
We can think of a matrix as a weighted adjacency matrix of some graph. A matrix
is then said to have a weighted edge if there is a edge in its graph whose value is strictly 
larger than 1.
Also, a matrix is said to have
two overlapping edges if there are two edges in the graph whose target node is the same.
Even though we were not able to resolve the decidability status for this family,
we identify a non-trivial sub-family of matrices in this family for which coverability is decidable.
This sub-family contains any class that contains only non-negative self-loop matrices,
i.e., non-negative matrices such that each diagonal entry is non-zero. For any such class,
we show that coverability is decidable and \NP-complete, and hence
no harder than coverability for continuous VASS.

\paragraph*{Computational Complexity. } We also complement all of our upper bounds with matching lower bounds, except for a single case: State-reachability for non-negative classes without any zero-rows or columns containing weighted/overlapping edges.

More precisely, we show that reachability is \NEXP-complete for classes containing only permutation matrices with at least one non-trivial permutation matrix i.e., a permutation matrix that is not the identity matrix.
We also show that coverability (resp. state-reachability) is \NEXP-complete (resp. \PSPACE-complete) for classes containing only permutation matrices
with at least one non-trivial permutation matrix and is \NP-complete for non-negative classes with only self-loop matrices.
Finally, we show that state-reachability is \PSPACE-complete for non-negative classes containing a matrix which has either a row or a column that is completely zero.

All of these results are formally defined and stated in the next section, with a summary in Table~\ref{table:results}.

\subsubsection*{Related work. } The closest work to ours is the paper~\cite{Affine-VASS}. 
In this paper, the authors consider reachability in affine VASS and affine integer VASS according to classes of affine operations and they completely classify 
the decidability of reachability for both these models across all classes. The latter model
is an over-approximation of affine VASS in which 
the counters are allowed to go below zero. The authors also prove a complexity trichotomy 
for affine integer VASS. In particular, the authors prove that reachability for integer VASS with resets
or transfers is decidable and mention that it is undecidable for the continuous case.
In this paper, we also use the formalization of classes used in~\cite{Affine-VASS}.

The paper~\cite{lmcs/BlondinHMR21} also discusses affine integer VASS. The authors of that paper
prove that when the set of matrices that are considered form a finite monoid, then reachability becomes
decidable for affine integer VASS. They also prove some other results for infinite monoids.


\subsubsection*{Our techniques. } Our results use a number of novel techniques that are different
from the ones used for affine VASS and affine integer VASS. Below, we explain some of our techniques which illustrate why many of  the results developed for affine VASS and affine integer VASS
are not applicable to affine continuous VASS. This also exhibits many surprising phenomena in the continuous case.

\paragraph*{Continuous VASS with resets. } As a first point of departure from the usual and integer semantics, we show that coverability for continuous VASS with resets is undecidable. This is surprising, because
in both the usual and the integer semantics, this problem is decidable. This illustrates
the complex nature of the continuous semantics - while it reduces the complexity for VASS,
for VASS with resets it actually makes it harder. To prove this undecidability result,
we first show that the 1-bounded coverability problem for continuous VASS with zero-tests is undecidable. In this problem, we are given a continuous VASS which is allowed to perform zero-tests and two of its configurations. 
We are then asked to check if the first configuration can cover the second one by a run
in which all the counter values always stay below 1. We show that this problem is undecidable
by observing that runs in the continuous semantics can be shrunk by arbitrary fractions.

We then reduce 1-bounded coverability for continuous VASS with zero-tests to coverability in reset VASS.
To do this, for every counter $x$ of the given continuous VASS with zero-tests, we introduce a new
\emph{complementary counter} $\overline{x}$ and maintain the invariant that the sum of the values
of $x$ and $\overline{x}$ is 1 throughout. Then we replace every zero-test with a reset. The idea
is that if the new machine resets some counter $x$ when it has a non-zero value, then the invariant
that $x + \overline{x} = 1$ is permanently destroyed and it instead becomes $x + \overline{x} < 1$.
Since the sum of $x$ and $\overline{x}$ in the final configuration will be 1, it would
then follow that the new machine will reset a counter only when it achieves the value 0, thereby
mimicking zero-tests.

\paragraph*{Reachability for weighted/overlapping edges. } While we do not resolve the 
status of coverability for classes with non-negative weighted/overlapping edges, we prove that reachability
for such classes is undecidable. 
A similar result is true for affine VASS, which can be proved by showing that such VASS can implement the doubling operation~\cite{Affine-VASS}. The reduction in this case crucially hinges on being able to double the current value of a counter and also being able to add 1 to a counter. While 
the former can also be mimicked in the continuous case, the latter cannot be, as the additive updates
are scaled with a non-deterministically chosen fraction $\alpha \in (0,1]$. It might 
be tempting to overcome this by adding a fresh counter which always has the value 1 and
then by using matrix operations, add this fresh counter's value whenever we want to add 1; however,
this operation might not be allowed within the class. 
We overcome this problem
by giving a reduction from the 1-bounded reachability problem for continuous VASS with zero-tests.
Given a continuous VASS with zero-tests, to every counter $x$, we attach a complementary counter $\overline{x}$ and a collection of other dummy counters. Further, we maintain the invariant that the sum of the values of $x$, $\overline{x}$ and the dummy counters is 1.
We then identify a matrix from this class such that when it is multiplied with a counter valuation in which $x$ and the dummy counters have value 0, then the invariant remains; otherwise 
the invariant is permanently broken. This then allows us to simulate zero-tests.

\paragraph*{State-reachability for non-negative matrices. } In affine VASS, state-reachability is inter-reducible with coverability, which is decidable for classes which contain only non-negative matrices, i.e., matrices in which all entries are non-negative.
This is because affine VASS with non-negative matrices are well-structured~\cite{wsts,wqts}. 
However, since non-negative rationals are not well-quasi ordered (under the usual ordering), 
we cannot apply the theory
of well-structured transition systems in our setting of affine continuous VASS.
Nevertheless, we show that 
the state-reachability problem is decidable and in \PSPACE \ for affine continuous VASS
for classes that contain only non-negative matrices. We do this by constructing a \emph{support
abstraction} graph from a given affine continuous VASS. The nodes of this graph only maintain the current
control state and the current set of counters which have a non-zero value. Edges
between the nodes of this graph are defined in a manner which satisfy some basic conditions: For instance,
if a transition $t$ increments some counter $x$ and if $(p,S) \act{t} (p',S')$ is an edge
in this graph labelled by $t$ then $S'$ must contain $x$. This is because whenever $t$ is fired
in the affine continuous VASS it will necessarily lead to a configuration with a non-zero value in $x$.
We then show that maintaining some such basic conditions is necessary and sufficient to decide
state-reachability and that this can be done in \PSPACE. As a particular example, this 
proves that state-reachability for continuous VASS with resets is decidable and proves that,
unlike the case of affine VASS,
coverability and state-reachability are two different problems for affine continuous VASS.

\paragraph*{Coverability for self-loop matrices. } Finally, to prove that coverability for non-negative self-loop matrices (i.e., non-negative matrices in which no diagonal entry is zero) is in \NP, we use results developed for continuous VASS. By adapting some techniques developed for continuous VASS, we show
that the coverability relation of an affine continuous VASS containing only non-negative self-loop matrices
can be defined as a formula in the existential theory of rationals equipped with the addition and the order operation (ELRA). The techniques developed for continuous VASS cannot be applied in a black-box manner 
and require some careful restructuring of the arguments so that it can be used for the affine case, tailored specifically to the case of coverability. This is because for any class of matrices
containing a self-loop matrix that is not the identity matrix, the reachability relation
cannot be reduced to formulas in the ELRA theory, since our results for reachability imply undecidability in this case.
This is in contrast to the case of continuous VASS 
for which the reachability relation can be embedded in ELRA.

To prove the decidability of coverability in self-loop matrices, we first identify counters
which could be potentially \emph{pumped}, i.e., these are counters which while taking some transition,
either double their old value or retain their old value and also get the value of another counter by means of the matrix of that transition.
Since pumpable counters only depend on the matrix of a transition and not on its additive update,
we can choose arbitrarily small fractions and still get big values for pumpable counters.
Subject to some mild technical constraints, we show how to modify a given run that starts and ends at the same state and pumps some counters,
into a run so that all pumped counters attain arbitrarily high values whilst not modifying the
values of the non-pumped counters. This is done by roughly splitting a run into many copies of the same run
with smaller fractions.
In the end, this would mean that with respect to coverability, we can get rid of the pumpable counters
and only focus on the non-pumpable ones. We then notice that the affine VASS with only the non-pumpable
counters is a continuous VASS, for which we know how to decide coverability.

\subsubsection*{Organization of the paper. }
The rest of the paper is organized as follows. In Section~\ref{sec:prelims} we introduce 
the notions of affine continuous VASS and classes of matrices and state our main results.
In Section~\ref{sec:undec} we prove all of our undecidability results. 
In Section~\ref{sec:dec} we prove all of our decidability and upper bound results.
In Section~\ref{sec:lower-bounds} we prove all of our lower bound results.
Finally we conclude in Section~\ref{sec:conclusion}.

All missing proofs can be found in the appendix of this paper.

\section{Preliminaries}\label{sec:prelims}

In this section, we define the notion of affine continuous VASS and the associated problems that we shall consider. Then, we define the notion of a \emph{class of matrices}, akin to~\cite{Affine-VASS}
and state our main results. 

Throughout this paper we let $\qnz$ denote the set of non-negative rational numbers
and $\qn_{>0}$ denote the set of positive rational numbers. We usually write
vectors $(u_1,\dots,u_d) \in \qn^d$ in bold as $\bu$ and we use $\bu(i)$ to 
denote $u_i$. Given some $\alpha \in \qn$ and a vector $\bv$, we let $\alpha \bv$
be the vector given by $(\alpha \bv)(i) = \alpha \cdot \bv(i)$.
The support of a vector $\bv$, denoted by $\supp{\bv}$ is the set $\{i : \bv(i) \neq 0\}$.
We use $\bzero$ to denote the vector that is 0 everywhere and we let $\bunit_i$ denote the
$i^{th}$ unit vector, i.e.,
$\bunit_i$ is 0 everywhere, except in the $i^{th}$ co-ordinate, where its value is 1.
Given two vectors $\bu, \bv$ we say that $\bu \le \bv$ if $\bu(i) \le \bv(i)$ for every $i$.

Given a matrix $\bA$ we use $\bA(i,j)$ to denote its entry in the $i^{th}$ row
and the $j^{th}$ column. Further, given a matrix $\bA$ and a vector $\bu$, we let 
$\bA \bu$ denote the vector obtained by multiplying $\bA$ with $\bu$.
Throughout this paper, we will only work with
square matrices, i.e., matrices whose number of rows and columns are the same.
A matrix $\bA$ is called a \emph{diagonal matrix} if 
all of its non-diagonal entries are 0, i.e.,
$\bA_{i,j} = 0$ whenever $i \neq j$. 
\subsubsection*{Affine continuous VASS. } Let $d \in \nn$. A $d$-affine continuous VASS (or simply affine continuous VASS) is a tuple $\mach = (Q,T)$ where $Q$ is a finite set of \emph{control states} and 
$T \subseteq Q \times \zn^{d \times d} \times \zn^{d} \times Q$ is a finite set of \emph{transitions}.
Note that the matrices and the vectors appearing in the transitions range over the integers.
Intuitively, $\mach$ has access to $d$ counters, each of which can store a non-negative rational value. 
The values of these counters can be manipulated by the transitions in $T$ as follows.

Given a transition $t = (p,\bA,\bb,q)$, we will call $\bA$ as the matrix of $t$ and $\bb$
as the additive update of $t$.
A \emph{configuration} of $\mach$ is a pair $(p,\bu)$ (mostly written as $p(\bu)$), where $p \in Q$ is a state and $\bu \in \qnz^d$
is a non-negative vector representing the current values of the $d$ counters. If $C = p(\bu)$
 is a configuration, then we let $C(i) = \bu(i)$ for any counter $i$.
Given two configurations
$p(\bu), q(\bv)$, a non-zero fraction $\alpha \in (0,1]$, and a transition $t$ of the form $t = (p,\bA,\bb,q)$, we say that there is a step from $p(\bu)$ to $q(\bv)$ by means 
of the pair $\alpha t$, denoted by 
$p(\bu) \act{\alpha t} q(\bv)$ if $\bv = \bA \bu + \alpha \bb$.
Intuitively, from the configuration $p(\bu)$, the machine executes the transition $t$ and goes to the state $q$ and updates the counters by multiplying the current vector $\bu$ with the matrix $\bA$ and then
adding the vector $\alpha \bb$.  Note that by the definition of a configuration, 
we implicitly require that $\bA \bu + \alpha \bb \ge \bzero$.
We say that $p(\bu) \act{} q(\bv)$, if there exists some non-zero fraction $\alpha$ and some transition $t$ such that $p(\bu) \act{\alpha t} q(\bv)$.

A \emph{firing sequence} $\sigma$ is a sequence in $((0,1] \times T)^*$. 
Given a firing sequence $\sigma = \alpha_1 t_1, \dots, \alpha_k t_k$ and a non-zero fraction $\beta$, we let $\beta \sigma$ denote the firing sequence $\beta \sigma = \beta \alpha_1 t_1, \dots, \beta \alpha_k t_k$.

Given two configurations
$C$ and $C'$ and a firing sequence $\sigma = \alpha_1 t_1, \dots, \alpha_k t_k$, we say that
$C \act{\sigma} C'$ if there exists configurations $C_1, \dots, C_{k-1}$ such that
$C \act{\alpha_1 t_1} C_1 \act{\alpha_2 t_2} C_2 \dots C_{k-1} \act{\alpha_k t_k} C'$.
Finally, we write $C \act{*} C'$ to mean that there is some firing sequence $\sigma$ such that
$C \act{\sigma} C'$.

Note that if all the matrices appearing in an affine continuous VASS $\mach$ are the identity matrices,
then we get the usual notion of continuous VASS~\cite{blondinLogicsContinuousReachability2017}.

\subsubsection*{Classes of matrices. } Our aim in this paper is to provide a classification of 
``classes of affine continuous VASS'' with respect to decidability of the reachability, coverability and state-reachability problems. To this end, we need to formalize the notion of a ``class of affine continuous VASS''. 
For this purpose, we use the framework used in~\cite{Affine-VASS}, which was used
for classifying \emph{affine VASS} and \emph{affine integer VASS} with respect to the reachability problem.

A class of affine VASS extends the usual notion of a VASS by providing some extra operations that can be applied to the counters by means of the affine transformations that it provides,
i.e., the set of allowed matrices in its transitions. Since affine VASS extend VASS, the identity matrix
must always be allowed. Also, the authors of~\cite{Affine-VASS} noticed that the set of allowed matrices must be closed under multiplication. Finally, a class of affine VASS typically does not pose any restrictions on the number of counters that can be used, or on the subset of counters to which the operations can be applied. This means that a matrix can be extended to arbitrary dimensions and can be applied to 
any subset of counters.  These observations are utilized in the definition of a class used in~\cite{Affine-VASS}. We now describe this formal defintion of a class and use it to classify affine continuous VASS as well.

For every $k \ge 1$, let $\bI_k$ be the $k \times k$ identity matrix. We use $\mathcal{S}_k$ to denote
the set of all permutations over $\{1, 2, \ldots, k\}$. For any permutation $\sigma \in \mathcal{S}_k$, let $\bP_{\sigma} \in \{0,1\}^{k \times k}$ be its corresponding matrix, i.e., for any $1 \le i, j \le k$, $\bP_{\sigma}(i,j)$ is 1 if $i$ is mapped to $j$ and 0 otherwise.
For any matrix $\bA \in \zn^{k \times k}$ and any permutation $\sigma \in \mathcal{S}_k$ we let 
$$\sigma(\bA) := \bP_\sigma \cdot \bA \cdot \bP_{\sigma^{-1}}$$
and for any $n \ge 1$, we let
$\ext{\bA}{n}$ be the $(k+n) \times (k+n)$ matrix defined as
$$\ext{\bA}{n} := \begin{pmatrix}
	\bA & \bzero_{\phantom{n}} \\
	\bzero & \bI_n
\end{pmatrix}$$

Intuitively, if $\bA$ is a $k \times k$ matrix that dictates the updates of some $k$ counters
of an affine continuous VASS, then $\sigma(\bA)$ essentially renames the $k$ counters according to the 
permutation $\sigma$, applies $\bA$ according to this renaming and then renames the counters back
to the original names. The operation $\ext{\bA}{n}$ can be thought of as adding $n$ fresh counters
and then applying the matrix $\bA$ only on the old $k$ counters, while leaving the new $n$ counters
unchanged.

Now, a \emph{class of matrices} $\cC$ is any subset of matrices in $\cup_{k \ge 1} \zn^{k \times k}$ that
satisfies the following four closure properties:
\begin{itemize}
	\item Identity: For every $n$, $\bI_n \in \cC$.
	\item Multiplication: If $\bA, \bB \in \cC$ are $k \times k$ matrices for some $k \ge 1$, then the matrix $\bA \cdot \bB$ is also in $\cC$.
	\item Extension: If $\bA \in \cC$, then for every $n$, the matrix $\ext{\bA}{n}$ is also in $\cC$.
	\item Counter renaming: If $\bA \in \cC$ is a $k \times k$ matrix, then $\sigma(\bA)$ is also in $\cC$ for any permutation $\sigma \in \mathcal{S}_k$.
\end{itemize}

We also define a new operation called Application: If $\bA$ is a $k \times k$ matrix and $i,n$ are numbers
such that $k \le n$ and $i+k-1 \le n$, then $\apply{n}{\bA}{i}$ is the $n \times n$ matrix defined as
$$\apply{n}{\bA}{i} := \begin{pmatrix}
	\bI_{i-1} & \bzero_{\phantom{n}} & \bzero \\
	\bzero & \bA & \bzero \\
	\bzero & \bzero & \bI_{n-i-k+1}
\end{pmatrix}$$

Intuitively, $\apply{n}{\bA}{i}$ applies the matrix $A$ to dimensions $i, i+1, \dots, i+k-1$ and leaves the other dimensions unchanged.
Note that $\apply{n}{\bA}{i}$ can be obtained by extending $\bA$ and renaming some counters of its extension: Indeed, $\apply{n}{\bA}{i} = \sigma(\ext{\bA}{n})$ where $\sigma$ is the permutation that maps
every $j$ to $i-1+j$ if $i-1+j \le n$ and otherwise to $(i-1+j) \bmod n$. Hence, Application is 
simply a syntactic sugar and all classes are closed under the Application operation. As we shall
see later on, the introduction of the Application operation will help us state us our constructions in an easier way.

We now show that the notion of a class of matrices can capture various sub-models of affine continuous VASS, each of which arise from continuous VASS by adding some new feature.

\begin{example}[Classes of matrices]
	As mentioned before, a continuous VASS is an affine continuous VASS in which all the matrices
	are identity matrices. This is captured by the \emph{identity class} of matrices, i.e.,
	the class that contains only the set of all identity matrices.
	For a $d$-continuous VASS $\mach = (Q,T)$ we will simply ignore the (identity) matrices in the transitions
	and denote $T$ by $T \subseteq Q \times \zn^d \times Q$.
	
	A continuous VASS with reset operations is a continuous VASS in which transitions are further 
	allowed to reset the value of a counter to 0. This model is captured by affine continuous VASS with
	matrices from the class $\cC_R$ which contains the set of all diagonal matrices with entries
	from $\{0,1\}$. Intuitively, if $\bA \in \cC_R$, and its $i^{th}$ diagonal entry is a 0, then this means that the $i^{th}$ counter
	is reset to 0 and if it is a 1, then the $i^{th}$ counter is left unchanged. 
	
	Another example is continuous VASS with transfer operations which extend continuous VASS
	with operations which allow the value of a counter $i$ to be transferred to a counter $j$. 
	This model is captured by affine
	continuous VASS with matrices from the class of all matrices 
	with entries from $\{0,1\}$ in which every column has at most one non-zero entry. 
\end{example}

\subsection{Contribution Part I - Decidability results}\label{subsec:contribution-part-I} 
We now state the main contributions of the paper.
Given a class $\cC$ of matrices, we say that an affine continuous VASS $\mach$ is a $\cC$-continuous VASS if all the matrices appearing in $\mach$ are in $\cC$. We are now
ready to define the problems of interest to us and state our main results.

The \emph{reachability problem} for a class $\cC$ of matrices is defined as the following:

\begin{quote}
	\noindent \textbf{Given: } A $\cC$-continuous VASS $\mach$ and two of its configurations $p(\bu), q(\bv)$.\\
	\noindent \textbf{Decide: } Whether $p(\bu)$ can reach $q(\bv)$.	
\end{quote}

Our first main result is a classification of the reachability problem, which completely
characterizes the decidability/undecidability border. To state this result, we set up a notation:
We say that a matrix $\bA$ is a \emph{permutation matrix} if $\bA = \bP_\sigma$ for some permutation $\sigma$. Our result for the reachability problem can then be phrased as the following.

\begin{theorem}\label{thm:main-reach}
	The reachability problem for a class $\cC$ is decidable (and in \NEXP) if $\cC$ contains
	only permutation matrices. Otherwise, the reachability problem is undecidable.
\end{theorem}

For our second result, we consider the \emph{state-reachability problem} for a class $\cC$ which is defined
as follows:

\begin{quote}
	\noindent \textbf{Given: } A $\cC$-continuous VASS $\mach$, a configuration $p(\bu)$ and a state $q$.\\
	\noindent \textbf{Decide: } Whether $p(\bu)$ can reach
	some configuration of the form $q(\bv')$ for any non-negative $\bv'$. 	
\end{quote}

Our second main result completely classifies the decidability/undecidability border for the 
state-reachability problem. To state this result, we set up a notation:
We say that a matrix $\bA$ is a \emph{non-negative matrix} if all of its entries are non-negative.
Our result for the state-reachability problem can then be phrased as the following.

\begin{theorem}\label{thm:main-state}
	The state-reachability problem for a class $\cC$ is decidable (and in \PSPACE) if $\cC$ contains
	only non-negative matrices. Otherwise, the state-reachability problem is undecidable.
\end{theorem}

Finally, we consider the \emph{coverability problem} for a class $\cC$ which is defined
as follows:

\begin{quote}
	\noindent \textbf{Given: } A $\cC$-continuous VASS $\mach$, and two of its configurations $p(\bu), q(\bv)$.\\
	\noindent \textbf{Decide: } Whether $p(\bu)$ can \emph{cover} $q(\bv)$, i.e., whether $p(\bu)$ can reach
	some configuration of the form $q(\bv')$ where $\bv' \ge \bv$.	
\end{quote}

Our third main result classifies the decidability/undecidability border for the 
coverability problem, except for a single family of classes. 
To state our result for this problem we need a notation.
A matrix is said to have a \emph{zero-row} (resp. \emph{zero-column}) if it contains
a row that is completely zero (resp. a column that is completely zero).
Our result for the coverability problem is then as follows.

\begin{theorem}\label{thm:main-cov}
	The coverability problem for a class $\cC$ is undecidable if $\cC$ contains
	a matrix with either a negative entry or a zero row or a zero column.
	It is decidable (and in \NEXP) if $\cC$ contains only permutation matrices.
\end{theorem}

To see how this result classifies the coverability problem except for a single family,
we need some definitions. 
A matrix $\bA$ is said to have a weighted edge if there exists
indices $i,j$ such that $\bA(i,j) > 1$. Further, $\bA$ is said to 
have two overlapping edges if there exists indices $i \neq j$ and $k$ such that
$\bA(i,k) > 0$ and $\bA(j,k) > 0$.
The intuitive idea is to think of $\bA$ as the weighted adjacency matrix of a graph.
Having a weighted edge in $\bA$ corresponds to having an edge in the graph whose value
is strictly bigger than 1. Having two overlapping edges in $\bA$ corresponds
to having two edges in the graph whose target nodes are the same.

Now, suppose $\bA$ is a non-negative matrix (over $\zn$) 
which does not contain a zero-row, a zero-column, a weighted edge or two overlapping edges.
Hence, for every index $\ell$, there is a unique index $\sigma(\ell)$ such that
$\bA(\sigma(\ell),\ell) = 1$ and $\bA(\ell',\ell) = 0$ for every $\ell' \neq \sigma(\ell)$.
We claim that $\sigma(i) \neq \sigma(j)$ whenever $i \neq j$. 
Suppose there exist $i \neq j$ such that $\sigma(i) = \sigma(j) = k$.
Then, since $\bA$ does not contain a zero-row, for every index $\ell$, there is some index $\tau(\ell)$ such that $\bA(\ell,\tau(\ell)) > 0$. Note that if $\ell' = \tau(\ell)$, then
$\sigma(\ell') = \ell$.

Now, starting from $i$, using $\tau$, we get 
a sequence of the form $\tau^0(i) := i, \ \tau^1(i) := \tau(i),\  \tau^2(i) := \tau(\tau(i)) \ldots$ Hence, there must be some $r$ such that $\tau^0(i),\cdots,\tau^{r}(i)$ are all distinct
and $\tau^{r+1}(i) = \tau^e(i)$ for some $0 \le e \le r$. 
If $e \neq 0$, then  $\bA(\tau^{e-1}(i),\tau^e(i)) > 0$ and $\bA(\tau^r(i),\tau^e(i)) > 0$, which is a contradiction because this leads to two overlapping edges.
Hence, $e = 0$ and so $\tau^r(i) = \sigma(i) = k$. Similarly, when this argument is applied to $j$, we get a similar sequence of indices $j, \tau(j), \cdots, \tau^w(j) = k$ for some $w$.

Now, by definition of $\sigma$ and $\tau$ we have that, $\sigma^{r}(k) = i$ and $\sigma^{w}(k) = j$.
Since $i \neq j$, we have $w \neq r$. Without loss of generality, let $w < r$. This means that $\sigma^{r-w}(j) = i$ and so $\tau^{r-w}(i) = j$. Since $\tau^r(i) = k$ and since $\bA(k,j) > 0$, we have $\bA(\tau^{r-w-1}(i),j) > 0$ and $\bA(\tau^r(i),j) > 0$,
which is a contradiction because this leads to two overlapping edges.
This means that $\sigma(i) \neq \sigma(j)$ whenever $i \neq j$. It can then be easily seen
that $\bA$ is a permutation matrix.

The discussion above then shows that our result for coverability covers all classes
except those that contain matrices with weighted/overlapping edges,
which we will call the weighted/overlapping family.
However, even within this family, we identify a subfamily for which the coverability problem is decidable.
We now describe this class of matrices.

A matrix $\bA$ is said to be a \emph{self-loop matrix} if no entry on the diagonal is zero. The intuition once again comes from graphs: If $\bA$ is a self-loop matrix, then 
in the graph of $\bA$, all nodes will have a self-loop.
We show that
\begin{theorem}\label{thm:cov-self-loop}
	The coverability problem for $\cC$ is decidable (and in \NP) if $\cC$ contains only
	non-negative self-loop matrices.
\end{theorem}

Note that \NP \ is the best possible upper bound possible for a class, 
since the coverability problem for continuous VASS is \NP-complete. Hence,
we have strictly extended the model of continuous VASS whilst retaining the same complexity
for the coverability problem.

\subsection{Contribution Part II - Complexity results} \label{subsec:contribution-part-II} In addition to the above results,
we complement all of our upper bounds with lower bounds, except for a single family:
The weighted/overlapping family. Our results
are then as follows.

\begin{theorem}\label{thm:complexity-reach-cov}
	The reachability and coverability problems for the class $\cC$ are \NEXP-complete if $\cC$ contains only permutation matrices with at least one non-trivial permutation	matrix.
	For the identity class both these problems are \NP-complete.
\end{theorem}

Here a non-trivial permutation matrix is a permutation matrix that is not the identity matrix.
Note that along with our decidability result for the reachability problem, this completely characterizes
the complexity of reachability for every class. 
We now move on to the state-reachability problem.

\begin{theorem}\label{thm:complexity-state}
	The state-reachability problem for $\cC$ is \PSPACE-complete, if $\cC$ contains
	only permutation matrices with at least one non-trivial permutation matrix
	or if $\cC$ contains a zero-row or a zero-column matrix.
	It is \NP-complete if $\cC$ is the identity class or $\cC$ is any class containing only
	self-loop matrices.
\end{theorem}

All of our results are summarized in Table~\ref{table:results}. We note that
the results for the identity class, i.e., continuous VASS, are already known~\cite{blondinLogicsContinuousReachability2017} and are only included here and in the previous
two theorems for the
sake of completeness.
\begin{center}
	\begin{table}[H]
		\caption{All of our results. The leftmost column contains types of matrices and the uppermost row contains the problems. An undecidability/hardness result for a cell means that the problem in that column is undecidable/hard for any class containing \emph{any} matrix in that row. A decidability/complexity upper bound
			result for a cell means that the problem in that column is decidable with the given upper bound for classes containing 
			\emph{no} matrices from any of the previous rows and containing
			\emph{only} those matrices in that row.}
	\begin{tabular}{|c|c|c|c|}
		\hline
		& \textbf{Reach} & \textbf{Cover} & \textbf{State-reach} \\ \hline
		\textbf{neg. entry matrices} & Undec. & Undec. & Undec. \\ \hline
		\textbf{zero-row/column} & Undec. & Undec. & \PSPACE-c \\ \hline
		\textbf{wt./overlap. edges} & Undec. & \multirow[t]{2}{*}{} $?$ & \multirow[t]{2}{*}{} \PSPACE \\ & & \NP-c ({\scriptsize self-loop}) & \NP-c ({\scriptsize self-loop}) \\ \hline  
		\textbf{non-trivial perm.} & \NEXP-c & \NEXP-c & \PSPACE-c \\ \hline
		\textbf{identity matrices} & \NP-c & \NP-c & \NP-c \\ \hline
	\end{tabular}
\end{table}\label{table:results}
\end{center}
\vspace{-20pt}
\paragraph*{Generic reductions between reachability, coverability and state-reachability. }
It is easy to see that state-reachability is a special case of coverability where
we want to cover the $\bzero$ vector. Further, coverability can be easily reduced to reachability:
Given an instance $(\mach,p(\bu),q(\bv))$ of the coverability problem, we modify $\mach$ to first add a transition which moves from $q$ to a new state $q'$. Then at $q'$ we add loops to decrement every possible counter. It then follows that $p(\bu)$ can cover $q(\bv)$ in $\mach$ if and only if 
$p(\bu)$ can reach $q'(\bv)$ in the new machine. Hence, we have

\begin{theorem}\label{thm:reach-cov-state}
	For any class $\cC$, the state-reachability problem for $\cC$ is poly.-time reducible
	to the coverability problem for $\cC$ which in turn is poly.-time reducible
	to the reachability problem for $\cC$.
\end{theorem}

We will use this theorem to simplify the presentation of many of our results.

\section{Undecidability results}\label{sec:undec}

In this section, we prove all of our undecidability results. 
We begin with our undecidability result on classes containing a matrix with a negative entry.

\subsection{State-reachability for classes with negative entries}

Our main result of this section is that

\begin{theorem}\label{thm:undec-negative-entries}
	The reachability, coverability and state-reachability problems for $\cC$-continuous VASS is undecidable
	if $\cC$ contains a matrix with some negative entry.
\end{theorem}

Note that by Theorem~\ref{thm:reach-cov-state}, it suffices to prove the above theorem
for only the state-reachability problem.
To prove this result, we introduce the notion of a continuous VASS with zero tests.
These are continuous VASS which can additionally test if the value of a counter is zero. 
Formally, a $d$-continuous VASS with zero tests is a tuple $\mach = (Q,T,T_{=0})$ where
$Q$ is a finite set of states, $T \subseteq Q \times \zn^d \times Q$ is a finite 
set of transitions and $T_{=0} \subseteq Q \times [d] \times Q$ is a finite set of zero-tests. 
Configurations of $\mach$ and steps of $\mach$ by transitions from $T$ are defined in the same way as for a continuous VASS. In addition, we also have steps defined by zero-tests in the following manner.
Given two configurations $p(\bu), q(\bv)$ and a zero-test $t$ of the form $t = (p,i,q)$, we say that there is a step from $p(\bu)$ to $q(\bv)$ by means of the zero-test $t$,
denoted by $p(\bu) \act{t} q(\bv)$ if $\bu(i) = 0$ and $\bv = \bu$. (Note that steps corresponding to a zero-test do not have any chosen fractions).
A run is then defined as a sequence of steps (defined by transitions and zero-tests) from one 
configuration to another. Reachability, coverability and state reachability are then defined in the same way as for a continuous VASS.
From~\cite[Theorem 4.17]{blondinLogicsContinuousReachability2017}, we have the following result.

\begin{theorem}\label{thm:cont-vass-zero-tests}
	The reachability, coverability and state-reachability problems for continuous VASS with zero-tests
	are undecidable.
\end{theorem}

Using this result, we shall prove Theorem~\ref{thm:undec-negative-entries}. 

\begin{proof}[Proof of Theorem~\ref{thm:undec-negative-entries}]
	Let $\cC$ be a class of matrices which has a matrix $\bA$ (of size $k \times k$ for some $k$)
	such that $\bA(i,j) < 0$ for some indices $i$ and $j$. First, we note the following property: For any $\lambda \in \qnz$,
	we have
	\begin{equation}\label{eq:one}
		\bA \cdot \lambda \bunit_j = \begin{cases}
			\bzero & \text{if } \lambda = 0\\
			\lambda \cdot \bA_j & \text{otherwise, where } \bA_j \text{ is the } j^{th} \text{ column of } \bA
		\end{cases}	
	\end{equation}
	
	This means that if we multiply $\bA$ by a non-negative vector $\bv$ which is zero everywhere, except possibly in the $j^{th}$ counter, 
	then the only way we can have a non-negative vector again is if 
	$\bv$ is $\bzero$. Intuitively, this means that if we have a configuration which is zero everywhere, except possibly 
	in the $j^{th}$ counter, then to test if the value of the $j^{th}$ counter is 0,
	we simply have to multiply the configuration with the matrix $\bA$. In this way, we can simulate
	a zero-test on counter $j$ by the matrix $\bA$. 
	
	We now show that $\cC$-continuous VASS can simulate continuous VASS with zero tests, which by Theorem~\ref{thm:cont-vass-zero-tests} will prove
	Theorem~\ref{thm:undec-negative-entries}. The reduction is similar to the one given in~\cite[Proposition 4.3]{Affine-VASS}.
	
	Let $\mach = (Q,T,T_{=0})$ be a $d$-continuous VASS with zero tests for some $d$.
	From $\mach$, we will create a $n = dk$ dimensional $\cC$-continuous VASS $\mach'$.
	For every $1 \le i \le d$, the $(j + (i-1)k)^{th}$ counter of $\mach'$ (which we will denote
	by $x_i$) will be responsible for simulating the $i^{th}$ counter of $\mach$.
	These counters, $x_1, x_2, \dots, x_d$, will be called the \emph{primary counters}
	of $\mach'$. 
%
	The remaining counters of $\mach'$ 
	will be called \emph{dummy counters} which will always have the value 0.
	
	Given a vector $\Delta \in \qn^d$, let ext$(\Delta) \in \qn^{n}$ be the vector which is 0 everywhere,
	except in counters $x_1, x_2, \dots, x_d$ where the values are respectively
	$\Delta(1),\Delta(2),\dots,\Delta(d)$. With this notation set up, we are ready to state
	the desired reduction. 
	
	The set of states of $\mach'$ will be the same as $Q$.  
	For each $t \in T \cup T_{=0}$ of $\mach$, 
	$\mach'$ will have a corresponding transition $t'$.
	Before we state this transition $t'$, we will state two properties that 
	we shall prove will be satisfied by $t'$.
	
	\begin{quote}
		\textsc{Property 1:} If 
		$(p,\text{ext}(\bu)) \act{\alpha t'} q(\bv')$ is a step in $\mach'$, then $\bv' = \text{ext}(\bv)$
		for some vector $\bv \in \qnz^d$.
	\end{quote}
	\begin{quote}
		\textsc{Property 2a):} 	If $t \in T$ is a transition of $\mach$,
		then $p(\bu) \act{\alpha t} q(\bv)$ is a step in $\mach$ if and only if 
		$p(\text{ext}(\bu)) \act{\alpha t'} q(\text{ext}(\bv))$ is a step in $\mach'$.
	\end{quote}
	\begin{quote}
		\textsc{Property 2b): } If $t \in T_{=0}$ is a zero-test of $\mach$,
		then $p(\bu) \act{t} q(\bv)$ is a step in $\mach$ if and only if 
		$p(\text{ext}(\bu)) \act{\alpha t'} q(\text{ext}(\bv))$ is a step in $\mach'$ for all $\alpha \in (0,1]$.
	\end{quote}

	We now proceed to state the transitions of $\mach'$.
	If $t = (p,\Delta,q)$ is a transition in $T$, then corresponding to $t$, we will have
	the transition $t' = (p,\bI_n,\text{ext}(\Delta),q)$ in $\mach'$.
	Note that $t'$ updates the counters $x_1, x_2, \dots, x_d$ of $\mach'$ in the same way
	as $t$ updates the counters $1,2,\dots,d$ of $\mach$. Further, $t'$ does not update the dummy counters
	of $\mach'$ at all. It can then be easily seen that $t'$ satisfies both the properties.
	
	If $t = (p,\ell,q)$ is a zero-test in $T_{=0}$, then corresponding to $t$, 
	we will have the transition $t' = (p,\apply{n}{\bA}{(\ell-1)k+1},\bzero,q)$ in $\mach'$.
	By definition of the Application operation, $t'$ does not update
	any counters in the set $\{1,2,\dots,(\ell-1)k\} \cup \{\ell k+1,\dots n\}$. Furthermore, by 
	equation~\ref{eq:one}, if at some configuration, the values of the counters in the set $\{(\ell-1)k+1,\dots,(\ell-1)k+j-1\} \cup \{(\ell-1)k+j+1,\dots,\ell k\}$ are zero, then the only way 
	$t'$ can be fired from that configuration is if the value of the counter $(\ell-1)k+j = x_\ell$ 
	is also zero.
	Further, firing $t'$ only results in all of the counters in the set $\{(\ell-1)k+1,(\ell-1)k+2,\dots,\ell k\}$ taking 
	the value 0.
	Hence, $t'$ also satisfies both of these properties.
	
	By using both the properties, it then follows that in $\mach$, a configuration $p(\bu)$
	can reach a configuration with state $q$ if and only if in $\mach'$, the configuration $p(\text{ext}(\bu))$
	can reach a configuration with state $q$. Hence, the desired reduction is complete.
\end{proof}

\subsection{Coverability for non-negative classes with zero-rows/columns}\label{subsec:zero-row-column}

Because of the negative result from the previous subsection, it suffices to only consider classes
which only contain non-negative matrices for the rest of the paper. Such a class will be called a non-negative class. We are now ready to state the main result of this subsection.

\begin{theorem}\label{thm:undec-zero-row-column}
	The reachability and coverability problems for $\cC$-continuous VASS is undecidable if $\cC$ is a non-negative class which contains
	a zero-row or a zero-column matrix.
\end{theorem}

By Theorem~\ref{thm:reach-cov-state}, it suffices to prove the above theorem only for the coverability problem. We do this in two stages.
In the first stage, we show that coverability for continuous VASS with resets, i.e., $\cC_R$-continuous VASS
where $\cC_R$ is the class of all diagonal matrices with entries from $\{0,1\}$ is undecidable.
In the second stage, we reduce from coverability for continuous VASS with resets to  
coverability for $\cC$-continuous VASS where $\cC$ contains a zero-row or a zero-column matrix.


\subsubsection*{Stage 1: Continuous VASS with resets. }

A continuous VASS with resets is simply an affine continuous VASS in which all
the matrices that appear are diagonal matrices with entries from $\{0,1\}$. 
For every $d,i$, let $\mathbf{R}^{(d,i)}$ be the $d\times d$ diagonal matrix
which has a 1 in each entry of the diagonal, except at the position $(i,i)$, where the value is 0.
Intuitively, if $\bu$ is some vector and $\bv = \mathbf{R}^{(d,i)}(\bu)$,
then $\bv$ is the same as $\bu$, except that $\bv(i) = 0$. The idea is that if $\bu$ represents
the current contents of all the counters, then $\mathbf{R}^{(d,i)}$ resets the value of the counter $i$
to 0 and does not alter the contents of the other counters. In the sequel,
whenever we say that a transition resets counter $i$, we mean that the current counter values
are multiplied by $\mathbf{R}^{(d,i)}$.
We will now show that
\begin{theorem}\label{thm:cont-vass-resets}
	The coverability problem for continuous VASS with resets
	is undecidable.
\end{theorem}

This theorem is surprising, because coverability in the usual VASS model enriched with resets
is decidable~\cite{wsts}. Hence, while the continuous semantics drastically reduces the expressive power 
in VASS, 
upon addition of resets,
the continuous semantics becomes even more powerful than the usual semantics.
%

To prove Theorem~\ref{thm:cont-vass-resets}, we first introduce the following notion.
Suppose $\mach$ is a continuous VASS \emph{with zero-tests} and $C$ and $C'$ are configurations of $\mach$.
A 1-bounded run between $C$ and $C'$ is a run in which all the counter values along the run are at most 1;
in particular this means that the counter values of $C$ and $C'$ are at most 1. The 1-bounded coverability problem for continuous VASS with zero-tests is the following: Given a tuple $(\mach,C,C')$, decide if there is a 1-bounded run from $C$ which can cover $C'$.
In order to prove Theorem~\ref{thm:cont-vass-resets}, we first prove the following lemma.
\begin{restatable}{lemma}{undeconebounded}\label{lem:undec-cov-zero-tests-0and1}
	The 1-bounded coverability problem is undecidable for continuous VASS with zero-tests.
\end{restatable}

The proof of this lemma can be found in Subsection~\ref{subsec:appendix-scaling-property}.
The intuitive idea behind this lemma is that continuous VASS with zero-tests have a \emph{scaling property},
which allows runs to be \emph{shrunk} by arbitrary fractions. More precisely, it is easy
to see that whenever $D \act{\alpha t} D'$ is a step in a continuous VASS with zero-tests
where $t$ is a transition, then so is $\beta D \act{\beta \alpha t} \beta D'$ for any $\beta > 0$
such that $\beta \cdot \alpha \le 1$. (The converse holds as well). This means
that runs and the configurations appearing in them can be shrunk down to very small values,
hence enabling us to reduce the coverability problem for continuous VASS with zero-tests
to its 1-bounded version.

Assuming that Lemma~\ref{lem:undec-cov-zero-tests-0and1} is true, we shall now prove Theorem~\ref{thm:cont-vass-resets}. 
Let $(\mach,C,D)$ be an instance of the 1-bounded coverability problem for 
continuous VASS with zero-tests.
To every counter $x$ of $\mach$, we add a \emph{complementary counter} $\overline{x}$ and modify the transitions of $\mach$ so that the sum of the values of $x$ and $\overline{x}$ is 1 at any point. 
This can be done by modifying each transition $t$ of $\mach$
so that if $t$ increments (resp. decrements) $x$ by some value $a$, then $t$ decrements (resp. increments) $\overline{x}$ by $a$ and vice versa. 
Then, we replace each zero-test of $\mach$ on the counter $x$, by a reset on the counter $x$.
Call this new machine $\mach'$. Note that $\mach'$ is a continuous VASS with resets.
Let $C'$ (resp. $D'$) be the configuration of $\mach'$ such that 
$C'(x) = C(x)$ and $C'(\overline{x}) = 1 - C(x)$ (resp. $D'(x) = D(x)$ and $D'(\overline{x}) = 1 - D(x)$)
for every counter $x$ of $\mach$. Note that since $(\mach,C,D)$ is an 
instance of the 1-bounded coverability problem, $C(x), D(x) \le 1$ for every $x$ and so $C'(\overline{x}), D'(\overline{x})$ are non-negative rationals.

If in any run of $\mach'$, a counter $x$ is reset when it already has the value 0, then this reset
does not disturb the invariant that the sum of the values of $x$ and $\overline{x}$ is 1.
On the other hand, if a run of $\mach'$ resets $x$ when it has a positive value $\alpha > 0$,
then this reset will make the sum of the values of $x$ and $\overline{x}$ to be $1 - \alpha$.
Since all the other transitions either reset a counter or maintain the sum of $x$ and $\overline{x}$ to be a constant, it would follow that such a run can never cover the configuration $D'$. 
Hence, in a run which covers $D'$ from $C'$ in $\mach'$, a counter $x$ is reset only when it has the value 0. Since we have replaced zero-tests in $\mach$ with resets in $\mach'$, it then 
follows that any run that covers $D$ from $C$ in $\mach$ can be uniquely mapped to a run
that covers $D'$ from $C'$ and vice-versa. This then completes the proof of Theorem~\ref{thm:cont-vass-resets}.

\subsubsection*{Stage 2: Reduction from resets to zero-row/zero-column matrices. }
We now show that affine classes that contain a matrix with a zero-row/column can simulate resets.
To this end, let $\cC$ be a non-negative class such that $\cC$ contains a matrix $\bA$ (of size $k \times k$ for some $k$) which either has a 
zero-row or a zero-column. Let $j$ be the row or column of $\bA$ which is completely zero.
Note that 
\begin{multline}\label{eq:two}
	\text{If $j$ is a row then for any $\bu \in \qnz^k$, }\bA \cdot \bu  = \bv \\ \text{ where } \bv \in \qnz^k \text{ satisfies } \bv(j) = 0
\end{multline}
and
\begin{equation}\label{eq:three}
	\text{If $j$ is a column then for any $\lambda \in \qnz$, } \bA \cdot \lambda \bunit_j = \bzero
\end{equation}

Equation~\ref{eq:two} tells us that we can use the matrix to always reset the counter $j$.
Equation~\ref{eq:three} tells us that as long as all the counters apart from $j$
on which we are applying $\bA$ have the value 0, then we can reset the counter $j$.
We can then exploit these equations to 
give a reduction from continuous VASS with resets to $\cC$-continuous VASS,
similar in style to the proof of Theorem~\ref{thm:undec-negative-entries}. 
For full details, we refer the reader to Section~\ref{subsec:appendix-stage2} of the appendix.

\subsection{Reachability for non-negative classes with weighted/overlapping edges}

We now move on to our final undecidability result. Our main result of this subsection is that

\begin{restatable}{theorem}{undecweight}~\label{thm:undec-weight-overlap-edges}
	The reachability problem for $\cC$-continuous VASS is undecidable if $\cC$ is a non-negative
	class which contains a matrix with a weighted edge or two overlapping edges.
\end{restatable}

The proof of this theorem will combine ideas from the previous two reductions. More specifically, we prove this theorem by giving a reduction from the 1-bounded reachability problem for continuous VASS with zero-tests. 
In this problem, we are given a continuous VASS with zero-tests $\mach$ and two configurations
$C,D$, and we have to check if $C$ can reach $D$ by a run in which all the counter values are at most 1 at any point.
In Lemma~\ref{lem:undec-cov-zero-tests-0and1}, we showed that the 1-bounded \emph{coverability} problem for continuous VASS with zero-tests is undecidable. By mimicking the transformation given in Theorem~\ref{thm:reach-cov-state}, 
it then follows that the 1-bounded reachability problem for continuous VASS with zero-tests is also undecidable.

\begin{proof}[Proof of Theorem~\ref{thm:undec-weight-overlap-edges}]
	Let $\cC$ be a non-negative class. We can assume that $\cC$ does not have any matrices with zero-rows/columns,
	as otherwise Theorem~\ref{thm:undec-zero-row-column} applies.
	Let $\bA$ be a matrix (of size $k \times k$ for some $k$) which either has a weighted edge
	or two overlapping edges. This implies that there is an index $z$ such that 
	$\sum_{1 \le i \le k} \bA(i,z) > 1$. This then allows us to deduce 
	the following equation. 
		
	\begin{equation}\label{eq:four}
		\bA \cdot \bu = \begin{cases}
			\bzero & \text{if } \bu = \bzero\\
			\bv & \text{with } \sum_{1 \le i \le k} \bv(i) > \sum_{1 \le i \le k} \bu(i) \text{ if } \bu(z) > 0
		\end{cases}	
	\end{equation}

	Intuitively, this equation gives us an ability to perform a zero-test on the $z^{th}$ counter.
	As in the reductions before, we can think of the $z^{th}$ counter as a primary counter
	and all the other $k-1$ counters as dummy counters holding the value 0.
	In our reduction below, we will pair the $z^{th}$ counter with a complementary counter $\overline{z}$
	such that any additive increase in the $z^{th}$ counter will be accompanied by an additive
	decrease in its complementary counter and vice-versa. This will ensure that the 
	sum of the values of $z, \overline{z}$ and the $k-1$ dummy counters
	will be 1 with respect to the additive updates.
	If at some point, the $z^{th}$ counter along with all the dummy counters have value 0, 
	then multiplying by 
	$\bA$ has no effect. On the other hand, if the $z^{th}$ counter has a non-zero value,
	then multiplying by $\bA$ has the effect of strictly increasing the sum of the values
	of these counters. 
	Since, the overall sum of the values
	of these counters in the final configuration will also be 1, 
	it follows that in any run between 
	the initial and the final configuration, we can multiply our current valuation with the matrix $\bA$ at some point
	only if the value of $z$ at that point is 0, which enables us to do a zero-test.
	Hence, we can then give a reduction from the 1-bounded reachability problem
	for continuous VASS with zero-tests in a style similar to Theorem~\ref{thm:undec-negative-entries}, which will prove the required
	undecidability result. 
	For full details, we refer the reader to Section~\ref{sec:appendix-weight-overlap-edges} of the appendix.	
	
	Note that this same reduction will \emph{not} work for coverability, because in coverability it is permissible to have strictly higher values than the target configuration.
\end{proof}

Let us now look at the big picture for all of our undecidability results.
We have shown all the undecidability results mentioned in Theorems~\ref{thm:main-state}
and~\ref{thm:main-cov}. As for the reachability problem, 
the discussion given in Subsection~\ref{subsec:contribution-part-I} tells
us that if $\cC$ contains a matrix $\bA$ that is not the
permutation matrix, then $\bA$ either has a negative entry, or has a zero-row/column,
or has weighted/overlapping edges. It then
follows that we have also proved all the undecidability results for the reachability problem mentioned in Theorem~\ref{thm:main-reach}.

\section{Decidability results and complexity upper bounds}\label{sec:dec}

We now prove all of our decidability results. 
We begin by showing decidability of state reachability for non-negative classes.

\subsection{State-reachability for non-negative classes}\label{subsec:dec-state-reach}

Our main result of this subsection is that
\begin{theorem}\label{thm:dec-state-reach}
	The state-reachability problem for $\cC$-continuous VASS is in \PSPACE \  if $\cC$ contains only non-negative matrices.
\end{theorem}

Let $\cC$ be a class containing only non-negative matrices and 
let $\mach = (Q,T)$ be a $d$-dimensional $\cC$-continuous VASS. 
For every counter $x$ and every transition $t = (p,\bA,\bb,q)$ of $\mach$,
we say that 
\begin{itemize}
	\item $x$ is additively incremented by $t$ if $\bb(x) > 0$.
	\item $x$ is additively decremented by $t$ if $\bb(x) < 0$. In this case,
	we let $\supp{t}^-_x$ be the set $\{y : \bA(x,y) > 0\}$.
\end{itemize} 
The intuition behind additive increments/decrements is that if $x$ is additively incremented/decremented
by $t$, then the additive update of $t$ increases/decreases the value of $x$.
The intuition behind $\supp{t}^-_x$ is that, 
if $x$ is additively decremented by $t$, then to fire $t$ at some point, 
there must be some counter $y \in \supp{t}^-_x$ 
which has a non-zero value at that point. 
We will expand upon these intuitions in the upcoming paragraphs, but first, we state 
the main object of our algorithm, namely the \emph{support abstraction}.

Corresponding to $\mach$, we associate an edge-labelled graph $SU_\mach$, called its support abstraction, as follows.
The set of vertices of $SU_\mach$ will be $Q \times 2^{\{1,\dots,d\}}$. Intuitively, each configuration $p(\bu)$ of $\mach$ corresponds to the vertex $(p,\supp{\bu})$ of $SU_\mach$. 
The set of edges of $SU_\mach$ (and its labels) are as follows. 
Suppose $t = (p,\bA,\bb,q)$ is a transition of $\mach$. Then, there exists an edge between vertices $(p,S)$ and $(q,S')$ labelled by $t$ if and only if 
for all counters $x$, 
\begin{itemize}
	\item If $x$ is additively decremented by $t$, then $\supp{t}^-_x \cap S \neq \emptyset$
	\item If $x$ is additively incremented by $t$, then $x \in S'$ and
	\item If $x \in S'$ then either $x$ is additively incremented by $t$ or $\supp{t}^-_x \cap S \neq \emptyset$.
\end{itemize}

The intuitive idea behind this edge relation is as follows: Suppose $p(\bu) \act{\alpha t} q(\bv)$
is a step in $\mach$ with $t = (p,\bA,\bb,q)$. Then $\bv = \bA \bu + \alpha \bb$. 
If $x$ is additively decremented by $t$ this means that $\alpha \bb(x) < 0$, and so
it must be the case that, $(\bA \bu)(x) > 0$. Since $\bA$ is a non-negative matrix, the latter condition is equivalent to 
$\supp{t}^-_x \cap \supp{\bu} \neq \emptyset$. 
Furthermore, notice that if $\bb(x) > 0$, then so is $\bv(x)$ and if $\bv(x) > 0$ then
either $(\bA \bu)(x) > 0$ or $\bb(x) > 0$. This argument, along with
the definition of edges in the support graph then allows us to conclude that
\begin{quote}
	\textsc{Soundness: } Suppose $p(\bu) \act{\alpha t} q(\bv)$ is a step in $\mach$.
	Then $(p,\supp{\bu}) \act{t} (q,\supp{\bv})$ is an edge in the support abstraction.
\end{quote}

Furthermore by analysing the construction of the support graph, we can show that 
\begin{quote}
	\textsc{Completeness: } Suppose $(p,S) \act{t} (q,S')$ is an edge in the support abstraction 
	$SU_\mach$
	and suppose $p(\bu)$ is a configuration of $\mach$ such that $S \subseteq \supp{\bu}$. Then, 
	there exists a configuration $q(\bv)$ and a fraction $\alpha$ such that $S' \subseteq \supp{\bv}$ and
	$p(\bv) \act{\alpha t} q(\bv)$ is a step in $\mach$.
\end{quote}

Combining these two properties along with a simple induction on the length of a run 
allows us to prove the following:
\begin{quote}
	$p(\bu)$ can reach a configuration whose state is $q$ in $\mach$ if and only if $(p,\supp{\bu})$
	can reach some vertex of the form $(q,S)$ in the graph $SU_\mach$.
\end{quote}

For more details behind the proof of correctness of these claims, we refer the reader to Section~\ref{sec:appendix-dec-state-reach}.

Hence, to decide state reachability in $\mach$, it suffices to explore the support abstraction graph. By standard techniques of constructing this graph on-the-fly in a non-deterministic manner,
we get a non-deterministic polynomial space algorithm for deciding state reachability of $\mach$.
Applying Savitch's theorem then gives us a \PSPACE \ algorithm for deciding state reachability.

\subsection{Coverability for non-negative classes with only self-loop matrices}\label{subsec:dec-cov}

The main result of this subsection is that 
\begin{theorem}\label{thm:dec-cov-self-loop}
	The coverability and state-reachability problems for $\cC$-continuous VASS is in NP \ if $\cC$ is a non-negative class
	that only contains self-loop matrices.
\end{theorem}

Note that by Theorem~\ref{thm:reach-cov-state}, we need only prove the above theorem for coverability.
We will prove this by capturing the coverability relation of such a $\cC$-continuous VASS as a formula in existential linear rational arithmetic (ELRA). Here ELRA is the existential theory of 
the rationals equipped with the addition and order operations.
Since satisfiability of this existential theory can be done in \NP~\cite{sontagRealAdditionPolynomial1985},
Theorem~\ref{thm:dec-cov-self-loop} would then follow.

We capture the coverability relation in ELRA by adapting some techniques developed for continuous VASS.  
Let us fix a class $\cC$ that contains only non-negative self-loop matrices and let us fix a  $d$-dimensional $\cC$-continuous VASS $\mach$ for the rest of this section. We first state a couple of basic facts
regarding $\mach$.

\begin{restatable}{lemma}{basicfacts}\label{lem:basic-facts}
	The following statements are true for any firing sequence $\pi$ and any fraction $\alpha \in (0,1]$:
	\begin{enumerate}
		\item\label{item:a} Suppose $p(\bu) \act{\pi} q(\bv)$. Then, $p(\alpha \bu) \act{\alpha \pi} q(\alpha \bv)$.
		\item\label{item:b} Suppose $p(\bu) \act{\pi} q(\bv)$. Then, for any $\bw \in \qnz^d$, 
		$p(\bu + \bw) \act{\pi} q(\bv')$
		where $\bv' \ge \bv + \bw$.
		\item\label{item:c} Suppose $p(\bu) \act{\pi} q(\bv)$. Then, $p(\bu) \act{\alpha \pi} q(\bv')$
		where $\bv' \ge \alpha \bv + (1-\alpha) \bu$.
	\end{enumerate}
\end{restatable}

Indeed in all these cases, the proofs are immediate when $\pi$ is a single step. 
The general case then follows by induction on the length of $\pi$.
A formal proof can be found in Section~\ref{subsec:appendix-basic-facts} of the appendix.

Now we move on to proving our main result, i.e., Theorem~\ref{thm:dec-cov-self-loop}.
To prove the result, first we introduce a notation. We say that an ELRA formula $\phi$ can be
computed in NP, if there is a Turing machine running
in non-deterministic polynomial time such that every accepting path of the machine
computes a \emph{disjunct} of the formula $\phi$. Hence, if 
$\phi_1,\dots,\phi_\ell$ are all of the produced formulas by all of the accepting paths,
then $\phi = \lor_{i=1}^\ell \phi_i$. Notice that if we can compute $\phi$ in NP, 
then we can test for satisfiability of $\phi$ in non-deterministic polynomial time
as well: Non-deterministically compute a disjunct $\phi_i$ of $\phi$ and then
check for emptiness of $\phi_i$, which can be done in \NP~\cite{sontagRealAdditionPolynomial1985}.

Equipped with this notion, we now prove Theorem~\ref{thm:dec-cov-self-loop}.
For this, notice that it suffices to show the following lemma.

\begin{lemma}\label{lem:cycle-cov}
	Given a state $p$ of $\mach$, in \NP, we can construct an ELRA formula 
	$\phi_p$ such that $\phi_p(\bu,\bv)$ is satisfiable if and only if 
	$p(\bu)$ can cover $p(\bv)$.
\end{lemma}

Indeed, by standard arguments on decomposing paths into steps and (non-simple) cycles~\cite{blondinLogicsContinuousReachability2017},
we can establish that a configuration $p(\bu)$ can cover a configuration $q(\bv)$ if and only if 
there is a \emph{witness}, i.e.,
a sequence of configurations $$p_1(\bu_1),p_1(\bu_1'),p_2(\bu_2),p_2(\bu_2'),
\dots,p_k(\bu_k),p_k(\bu_k')$$ such that each $p_i \neq p_j$ for $i \neq j$,
$p_1(\bu_1) = p(\bu), p_k = q, \bu_k' \ge \bv$ and for each $i$, $p_i(\bu_i)$ can cover $p_i(\bu_i')$
and $p_i(\bu_i') \act{} p_{i+1}(\bu_{i+1})$. Since we can 
we can easily come up with a formula $\phi_{p,q}^{s}$ in ELRA
such that $\phi_{p,q}^s(\bu,\bv)$ is true if and only if $p(\bu) \act{} q(\bv)$, i.e.,
$p(\bu)$ can reach $q(\bv)$ by a single step, Lemma~\ref{lem:cycle-cov} then
allows us to solve coverability by reducing it to formulas in ELRA, which can
be checked for satisfiabilty in \NP~\cite{sontagRealAdditionPolynomial1985}.
More details can be found in Subsection~\ref{subsec:appendix-cyclic-to-general}.

\subsubsection*{Proving Lemma~\ref{lem:cycle-cov}. }

Hence, all that remains is to prove Lemma~\ref{lem:cycle-cov}. 
To prove this, we
introduce some definitions, then give a high-level overview of our proof strategy
and then prove a number of preparatory results leading to the main result.

Recall that we have fixed a class $\cC$ that contains only non-negative self-loop matrices and a $d$-dimensional $\cC$-continuous VASS $\mach$ for the rest of this section. 
Given a transition $t = (p,\bA,\bb,q)$ of $\mach$ and a counter $x$,
we had already defined the notions of $x$ being additively incremented by $t$, decremented by $t$ and $\supp{t}^-_x$ in 
Subsection~\ref{subsec:dec-state-reach}. Here, we introduce another notion called pumping.
We say that a counter $x$ can be pumped by transition $t$ using $y$ if either $y = x$ and $\bA(x,y) > 1$ or $y \neq x$ and $\bA(x,y) > 0$. We then let $\supp{t}^p_x = \{y : (y \neq x \land \bA(x,y) > 0) \lor (y = x \land \bA(x,y) > 1)\}$.	If $x$ can be pumped by $t$ by using counter $y$, then since 
the matrix of $t$ is a self-loop matrix, once we fire $t$, the new value of $x$
will have the sum of the old values of $x$ and $y$ (and also the additive changes done by $\bb$).
Moreover, $\supp{t}^p_x$ contains all the counters with which $t$ can pump $x$.

Given a run between $p(\bu)$ and $q(\bv)$ and a counter $x$, we say that counter $x$ was \emph{pumped} in this run if there exists a step $p'(\bu') \act{\alpha t} q'(\bv')$
in this run
such that $x$ can be pumped by $t$ using $y$ and $\mathbf{\bu'(y)} > 0$.
As we shall see later on, counters that can be pumped in a run can be made to have arbitrarily high values
subject to certain mild restrictions on the other counters.


With these notations set up, we now give a high-level overview behind the proof of Lemma~\ref{lem:cycle-cov}. First, we prove a lemma called the Pumping Counters lemma which states that
if a counter $x$ is pumped in a run of the form $p(\bu) \act{*} p(\bv)$, then (subject to some technical constraints), we can get another run from $p(\bu)$ in which the final counter value of $x$ can be made
arbitrarily high and all the other counter values are at least as much as their respective values in $\bv$.
Hence, the Pumping Counters lemma allows us to \emph{pump} the value of $x$ whilst maintaining the values
of the other counters as well. Hence, for the purposes of coverability, we can throw away all of these 
pumpable counters. As for the remaining counters,
notice that the affine continuous VASS with respect to this set of counters behaves exactly like a continuous VASS, for which
we know how to decide coverability. Hence, intuitively the overall idea is to then identify pumpable counters, remove
them from $\mach$ and only concentrate on the machine with the remaining counters which behaves like a continuous VASS and which we know how to solve. This reasoning is proved in another lemma which 
we call the Coverability Characterization lemma. It then suffices to translate the conditions
given in this lemma to a formula in ELRA, which we do in order to prove Lemma~\ref{lem:cycle-cov}.

Having presented the high-level idea, we now begin the proof of Lemma~\ref{lem:cycle-cov} by 
first proving a number of preparatory results.
Our first such result states that whenever there is a run from a state $p$ back to itself, then 
by following the same path with different firing fractions, we can reach a configuration which is non-zero on every counter that was ``affected'' by the run.
This proposition is similar to~\cite[Lemma 4.4]{blondinLogicsContinuousReachability2017}.

\newcommand{\rephalf}{\text{Rep-Half}}
\begin{restatable}{proposition}{supportprop}\label{prop:support-prop}
	Suppose we have a run $p(\bu) \act{\pi} p(\bv)$ for some $p, \bu, \bv$ and $\pi$.
	Then, there is a firing sequence $\rephalf(\pi)$ such that 
	\begin{itemize}
		\item If $\pi = \alpha_1 t_1, \alpha_2 t_2, \dots, \alpha_\ell t_\ell$, then
		$$\rephalf(\pi) = \alpha_1/2 t_1, \alpha_2/4 t_2, \dots, \alpha_\ell/2^\ell t_\ell$$
		\item $p(\bu) \act{\rephalf(\pi)} p(\bw)$ for some $\bw$ such that $\supp{\bu} \subseteq \supp{\bw}$, and
		\item If $t$ is a transition in $\pi$ and $x$ is some counter, then
		\begin{itemize}
			\item If $x$ is additively incremented by $t$, then $\bw(x) > 0$.
			\item If $x$ is additively decremented by $t$, then $\supp{t}^-_x \cap \supp{\bw} \neq \emptyset$.
			\item If $x$ was pumped by $t$ in the run $p(\bu) \act{\pi} p(\bv)$  , then $\supp{t}^p_x \cap \supp{\bw} \neq \emptyset$.
		\end{itemize}
\end{itemize}
\end{restatable}

The intuitive idea is that in the new path
we repeatedly halve the firing fraction at each step. 
$\rephalf(\pi)$ now has the property that we leave a very small ``residue''
of every intermediate configuration that we visited by firing $\pi$ from $p(\bu)$.
More details can be found in Section~\ref{subsec:appendix-support-prop} of the appendix.

We now state a main technical lemma, namely the Pumping Counters Lemma.
It states that whenever we have a run that pumps some counters,
then, subject to some mild constraints, that run can be modified to arrive at a configuration
in which all the pumped counters have an arbitrarily large value.

\begin{restatable}[Pumping Counters Lemma]{lemma}{pumpingcounters}~\label{lem:pumping-counters}
	Suppose $\rho: p(\bu) \act{\pi} p(\bv)$ is a run.
	Let $S$ be the set of counters such that $x \in S$
	only if $x$ is pumped by some transition $t$ in $\rho$
	and $\supp{t}^p_x \cap \supp{\bv} \neq \emptyset$.
	Then, for any $K > 0$, there exists a run $\pi'$ such that
	\begin{itemize}
		\item $\pi' = \pi/2 \cdot (\pi/2n)^n$ for some $n$.
		\item $p(\bu) \act{\pi'} p(\bw)$ such that $\bw \ge \bv$ and
		\item For every $x \in S$, $\bw(x) > K$.
	\end{itemize}	
\end{restatable}

The intuitive idea is behind the proof is as follows: We first fire $\pi/2$ from $p(\bu)$. 
By Lemma~\ref{lem:basic-facts}~(\ref{item:c}), this has
the effect of reaching a configuration of the form $p(\bu_1)$ with $\bu_1 \ge \bv/2 + \bu/2$.
By our assumption on $S$, this now means that for every counter $x \in S$ 
that was pumped by some $t$, there
is some counter $c_x$ such that $c_x \in \supp{t}^p_x \cap \supp{\bu_1}$.
Now, from $p(\bu_1)$ we fire $\pi/2n$ to reach $p(\bu_2)$. 
The effect of doing this is that now for every counter $x \in S$,
the value of $x$ in $\bu_2$ is roughly at least $\bu_1(x) + \bv(c_x)/2$.
Now, from $p(\bu_2)$ we fire $\pi/2n$ again to reach $p(\bu_3)$.
Once again the effect of doing this is that for every counter $x \in S$,
its new value is roughly at least $\bu_2(x) + \bv(c_x)/2$ which is roughly
$\bu_1(x) + 2\bv(c_x)/2$. By iterating this procedure for $n$ times,
every counter $x \in S$ will roughly at least have the value $\bu_1(x) + n\bv(c_x)/2$.
Choosing $n$ to be large enough will then make this value bigger than $K$.
More details can be found in Subsection~\ref{subsec:appendix-pumping-counters} of the appendix.

We will now use the Pumping Counters lemma to prove a characterization of coverability in $\mach$. 
To state this characterization we first need some notation. Given a firing sequence $\pi$
and a configuration $p(\bu)$, we say that $\pi$ is $p$-admissible from $p(\bu)$
if $p(\bu) \act{\pi} p(\bv)$ for some configuration $p(\bv)$.
Given a set of counters $X$ and a vector $\bv$ we use $\bv_X$ to denote the projection of $\bv$
onto the counters $X$.
Further, we let $\mach_X$ be the $|X|$-dimensional \emph{continuous VASS} obtained from $\mach$,
by converting each transition $t = (p,\bA,\bb,q)$ into a transition $t_X = (p,\bI,\bb_X,q)$
Given any firing sequence $\pi = \alpha_1 t_1, \dots, \alpha_k t_k$
over $\mach$,
we will let $\pi_X$ be the firing sequence $\pi_X = \alpha_1 t_{1_X}, \dots, \alpha_k t_{k_X}$
in $\mach_X$. Note that every firing sequence in $\mach_X$ is of the form $\pi_X$ for some firing 
sequence $\pi$ in $\mach$. With these notations set up, we can now state our characterization of coverability in $\mach$.


%
%
%
%
%

\newcommand{\pitfwd}{\pi_{\textsf{fwd}}}
\newcommand{\pitbwd}{\pi_{\textsf{bwd}}}
\newcommand{\bs}{\mathbf{s}}

\begin{restatable}[Coverability Characterization Lemma]{lemma}{covcharac}\label{lem:cov-charac}
	$p(\bu)$ can cover $p(\bv)$ in $\mach$
	if and only if there exists a set of counters $X$, a vector $\bw \ge \bv$ and firing sequences $\pi',\pi_{\textsf{fwd}}$ over $\mach$ such that
	\begin{enumerate}
		\item $p(\bu_X) \act{\pi'_X} p(\bw_X)$ is a run in $\mach_X$
		\item $\pi_{\textsf{fwd}}$ is $p$-admissible from $p(\bu)$ in $\mach$
		\item Every counter $x$ that is \textbf{not in} $X$ is pumped by the $p$-admissible run obtained from $p(\bu)$
		by firing $\pi_{\textsf{fwd}}$
		\item $\supp{\pi'} = \supp{\pi_{\textsf{fwd}}}$
	\end{enumerate}
\end{restatable}

If $p(\bu)$ can cover $p(\bv)$ in $\mach$ by a run of the form $p(\bu) \act{\pi} p(\bw)$ with $\bw \ge \bv$,
then setting $\pi' = \pitfwd = \pi$ and $X$ to be the set of counters not pumped in $\pi$ satisfies all the required conditions.


On the other hand, suppose we have $X, \bw, \pi'$ and $\pitfwd$ satisfying all the conditions of the lemma.
Intuitively, the counters that are in $X$ are roughly the counters that are not pumped.
For these counters, we can simply ignore the matrices and concentrate only on the underlying continuous VASS $\mach_X$. By assumption, we have the run $p(\bu_X) \act{\pi'_X} p(\bw_X)$ in $\mach_X$.
By using results about continuous VASS, we show that there is a small enough fraction $\lambda$ and a configuration $p(\bz')$ in $\mach_X$
such that $p(\bu_X) \act{\lambda \rephalf\left(\pi_{\text{fwd}_X}\right)} p(\bz')$ 
and  $p(\bz')$ can reach $p(\bw_X)$ by a firing sequence $\eta_X$ in $\mach_X$. 
This means that if in the original machine $\mach$, we have a vector $\bs$ such that
$\bs(x) \ge \bz'(x)$ for every $x \in X$ and $\bs(x)$ is very high for every $x \notin X$,
then we can cover $p(\bw)$ from $p(\bs)$ in $\mach$, by simply mimicking $\eta_X$ in $\mach$.

By then using the Pumping Counters lemma and the firing sequence $\pitfwd$, we then show that such a vector $\bs$ exists and also that $p(\bs)$ can be 
reached from $p(\bu)$ in $\mach$. This means that to cover $p(\bv)$ from $p(\bu)$, we first
move from $p(\bu)$ to $p(\bs)$ and then from there, cover $p(\bw)$ as mentioned in the previous paragraph.
Since $\bw \ge \bv$, this would give the required run from $p(\bu)$.  
The formal details behind this proof can be found in Subsection~\ref{subsec:appendix-cov-charac} in the appendix.


%

%

Hence, to embed the coverability relation as a formula $\phi_p(\bu,\bv)$ of ELRA, 
it suffices to express the conditions given by the Coverability 
Characterization lemma. We will now give a non-deterministic polynomial-time algorithm which does 
that.
Our algorithm will first guess a set of counters $X$ and a set of transitions $S$.
Then, it will write a formula of the form $\exists \bw \phi_1(\bu_X,\bw_X) \land \phi_2(\bu)$.
The formula $\phi_1(\bu_X,\bw_X)$ will express that there is a run from $\bu_X$
to $\bw_X$ in the continuous VASS $\mach_X$ whose support is $S$. Thanks to results
from continuous VASS~\cite[Lemma 4.8]{blondinLogicsContinuousReachability2017}, it is known how to construct $\phi_1(\bu_X,\bw_X)$ in polynomial-time.
The formula $\phi_2(\bu)$ will express that there is a firing sequence $\pitfwd$
whose support is $S$, which is $p$-admissible from $p(\bu)$ in $\mach$,
and the run resulting from $p(\bu)$ and $\pitfwd$ pumps all the counters not in $X$, i.e.,
all the counters in the complement of $X$. 

To construct $\phi_2$, we first give a characterization of the existence of admissible runs
which pump a given set of counters $Y$. To state this characterization, we need some notations.
The graph of $\mach$, denoted by $\gr_\mach$ is the graph whose nodes are the states of $\mach$
and there is an edge between $p$ and $q$ if and only if there is a transition $t$ such that 
$t = (p,\bA,\bb,q)$. In this case, we also let $in(t)$ be $p$, $out(t)$ be $q$ and $\supp{t}^+ = \{x : \bb(x) > 0\}$, i.e., $\supp{t}^+$ is the set of counters additively incremented by $t$.
We then have the following lemma.

\begin{restatable}{lemma}{far}\label{lem:far}
	There exists a $p$-admissible run from $p(\bu)$ which uses exactly the transitions
	from $S$ and which pumps all the counters in a set $Y$ if and only if 
	there exists an injection $f: S \to \qn_{> 0}$ 
	such that for every $t, t' \in S$,
	\begin{itemize}
		\item If $f(t)$ has the minimum value among all transitions in $S$, then $in(t) = p$.
		\item If $f(t) < f(t')$, then there exists a path in $\gr_\mach$ from
		$out(t)$ to $in(t')$ which uses only the transitions from the set $\{s: s \in S, f(s) < f(t')\}$.
		\item If $f(t)$ has the maximum value among all transitions in $S$, then there exists a path in $\gr_\mach$ from
		$out(t)$ to $p$ which uses only the transitions from $S$.
		\item For every counter $y$ additively decremented by $t$,
		$$\supp{t}^-_y \bigcap \left(\supp{\bu} \cup \bigcup\limits_{s \in S, f(s) < f(t)} \supp{s}^+\right) \neq \emptyset$$
		\item For every $y \in Y$, there exists a transition $r \in S$ such that 
		$$\supp{r}^p_y \bigcap \left(\supp{\bu} \cup \bigcup\limits_{s \in S} \supp{s}^+\right) \neq \emptyset$$
	\end{itemize} 
\end{restatable}

The intuition behind this lemma is that $f$ characterizes the \emph{first order appearance}
of the transitions that appear in a $p$-admissible run, i.e., if $f(t) < f(t')$, then $t$
appears before $t'$ in the $p$-admissible run. The proof of this lemma is very similar to
the one from~\cite[Lemma 4.7]{blondinLogicsContinuousReachability2017}.
and can be found in Subsection~\ref{subsec:appendix-far} of the appendix.

Hence to construct $\phi_2$ (given $p, S$ and the complement of $X$), 
it is sufficient to express these five conditions in ELRA. In~\cite[Lemma 4.8]{blondinLogicsContinuousReachability2017}, 
it is shown how to express the first four conditions in ELRA. 
It is also easy to see that the fifth condition can also be expressed in ELRA . This then completes 
the proof of Lemma~\ref{lem:cycle-cov} and hence also Theorem~\ref{thm:dec-cov-self-loop}.

\subsection{Reachability for non-negative classes with only permutation matrices}\label{subsec:dec-reach-perm}

We now show our final decidability result.
\begin{theorem}\label{thm:dec-perm-reach}
	The reachability and coverability problems for $\cC$-continuous VASS are in \NEXP \
	if $\cC$ contains only permutation matrices.
\end{theorem}

By Theorem~\ref{thm:reach-cov-state}, it suffices to prove the above theorem for only
the reachability problem.
We prove this by giving a reduction to the reachability problem for continuous VASS. This reduction is 
exactly the same as the one given in~\cite[Proposition 4.1]{Affine-VASS}, which gives a similar reduction in the case of affine VASS. The intuitive idea is that since $\cC$ contains only
permutation matrices, essentially each affine operation of $\mach$ is only a renaming 
of the counters. Hence, at each point, we can keep track of the current renaming in the
control states, which allows us to reduce the reachability problem for $\cC$-continuous VASS
to reachability in (an exponentially) bigger continuous VASS, which is in \NP.
More formal details can be found in Section~\ref{sec:appendix-dec-reach-perm} of the appendix.

Let us now review all of our decidability results and see how they prove the decidability and complexity upper bound claims made in Section~\ref{sec:prelims}. Note that the \NP \ upper bounds for the
identity class are already known~\cite{blondinLogicsContinuousReachability2017}. With this in mind,
it is then easy to check that the three decidability (and upper bound) results
proved in this section prove all of the decidability (and upper bound) claims in
Theorems~\ref{thm:main-reach},~\ref{thm:main-state}, ~\ref{thm:main-cov},~\ref{thm:cov-self-loop},~\ref{thm:complexity-reach-cov} and~\ref{thm:complexity-state}.

\section{Complexity lower bounds}\label{sec:lower-bounds}

We now prove complexity lower bounds for almost all of the upper bounds that
we proved in the previous section.  First, note that since any class containing self-loop matrices
generalizes the identity class, the coverability and state-reachability problems are \NP-hard
for self-loop classes.
We now move on to the non-trivial lower bounds.

\subsection{State-reachability for non-negative classes with zero-rows/columns}\label{subsec:lower-bound-zero-row-column}

In this subsection, we show that
\begin{theorem}
	The state-reachability problem for $\cC$-continuous VASS is \PSPACE-hard if $\cC$ contains a non-negative
	matrix with a zero-row or a zero-column.
\end{theorem}

We show that this theorem is true for continuous VASS with resets.
The general case then follows because the reduction given for coverability in Stage 2 in Subsection~\ref{subsec:zero-row-column} is also valid for state-reachability. 

To prove that state-reachability for continuous VASS with resets is PSPACE-hard,
we give a reduction from state-reachability for Boolean programs, for which state-reachability is 
PSPACE-hard~\cite{tacas/GodefroidY13}. Intuitively, a Boolean program
is a finite-state automaton which has access to $d$ Boolean variables (for some $d$), which can either
be set to 0 or 1 and which can be tested for its current value. 
Given a Boolean program with $d$ variables, we will construct a $2d$-continuous VASS with resets
that will simulate the Boolean program. The intuitive idea is that each Boolean variable 
will be simulated by two counters. Exactly one of these two counters 
will have a non-zero value at some point. If the first counter (resp. second) has a non-zero value,
then this corresponds to the Boolean variable being true (resp. false). Setting a Boolean variable to 0 or 1 can be done by resetting one of the counters and incrementing the other.
Testing the current value of a Boolean variable is accomplished by trying to decrement the appropriate counter. 
For a formal proof, we refer the reader to Section~\ref{sec:appendix-lower-bound-zero-row-column}
of the appendix.

\subsection{Non-negative classes with only permutation matrices}\label{subsec:lower-bound-perm}

In this subsection, we show that
\begin{theorem}
	If $\cC$ contains a non-trivial permutation matrix, then state-reachability is \PSPACE-hard and reachability and coverability
	are \NEXP-hard.
\end{theorem}

To this end, let $\cC$ contain
a non-trivial permutation matrix $\bP_\sigma$.
Since $\bP_\sigma$ is not the identity matrix,
there exists a sequence of distinct indices $i_1, i_2, \dots, i_\ell$ with $\ell \ge 2$ such that $\sigma$ maps $i_1$ to $i_2$, $i_2$ to $i_3$ and so on and finally maps $i_\ell$ to $i_1$.
Let $z = i_1$. Further, let $n \ge 2$ be the least number such that $\bP_\sigma^n = \bI$. Note that since
$\bP_\sigma$ is a permutation matrix, such a number must exist. 

Now, suppose $\bu$ is a vector such that $\bu(j) = 0$ for every $j \neq z$. Then, 
$\bP_\sigma \cdot \bu = \bv$ where $\bv(j) = 0$ for every $j \neq \sigma(z)$.
Similarly, suppose $\bu$ is a vector such that $\bu(j) = 0$ for every $j \neq \sigma(z)$.
Then, $\bP_\sigma^{n-1} \cdot \bu = \bv$ where $\bv(j) = 0$ for every $j \neq z$.
Intuitively, $\bP_\sigma$ transfers the value of $z$ to $\sigma(z)$ and does nothing to the other counters.
Similarly, $\bP_\sigma^{n-1}$ 
transfers the value of $\sigma(z)$ to $z$ and does nothing to the other counters.
Hence, just like in the previous subsection,
we can use $z$ and $\sigma(z)$ to encode a Boolean variable. This lets us reduce 
state-reachability for Boolean programs to state-reachability for $\cC$-continuous VASS,
Further, the same reduction allows us to reduce \emph{continuous Boolean programs} to 
$\cC$-continuous VASS, for which the reachability and coverability
problems are \NEXP-hard~\cite{pacmpl/BalasubramanianMTZ24}. 
The formal proofs can be found in Section~\ref{sec:appendix-lower-bound-perm} of the appendix.

Finally, we note that the results proved in this subsection cover all the lower bounds
mentioned in Theorems~\ref{thm:complexity-reach-cov} and~\ref{thm:complexity-state}.

\section{Conclusion}\label{sec:conclusion}

We have introduced the model of affine continuous VASS and studied the reachability, state-reachability 
and coverability problems for different classes of integer affine operations. We have shown a classification
for decidability of reachability and state-reachability and an almost-complete classification for coverability.
We have also complemented most of our decidability results with tight complexity bounds.

As part of future work, we would like to completely classify the decidability status of the coverability problem by solving it for the weighted/overlapping family. Extending the current approach for self-loop classes to this general case is challenging. 
This is because for the self-loop classes, we crucially exploited the fact that at each step, the matrix in the underlying transition does not allow the counter values to drop. 
This property is however lost in the general case, which indicates that different methods might be necessary to solve the general case.

Finally, it might be worth studying affine operations in which the matrices
are allowed to have entries from the rationals. Many of our undecidability and decidability
results already apply in this case as well; however, a complete classification of problems in this case
is a possible direction for future work.

\label{beforebibliography}
\newoutputstream{pages}
\openoutputfile{main.pages.ctr}{pages}
\addtostream{pages}{\getpagerefnumber{beforebibliography}}
\closeoutputstream{pages}
\bibliography{refs}

\newpage
\appendix
\section{Proofs of Subsection~\ref{subsec:zero-row-column}}\label{sec:appendix-zero-row-column}

\subsection{Proof of Lemma~\ref{lem:undec-cov-zero-tests-0and1}}\label{subsec:appendix-scaling-property}

\undeconebounded*
\begin{proof}
	We prove this lemma by giving a series of reductions that starts at the coverability problem
	for continuous VASS with zero-tests and ends at the 1-bounded coverability problem
	for continuous VASS with zero-tests.
	
	Let $\mach$ be a continuous VASS with zero-tests. 
	By definition of steps in $\mach$, note that the following is true:
	\begin{quote}
		\textsc{Scaling Property: } Let $D$ and $D'$ be configurations of $\mach$. 
		For any transition $t$, any $\alpha \in (0,1]$ and any $\beta > 0$ such that $\beta \cdot \alpha \le 1$,
		we have $D \act{\alpha t} D'$ if and only if $\beta D \act{\beta \cdot \alpha t} \beta D'$.
		Further, for any zero-test $t$, we have $D \act{t} D'$ if and only if for any $\beta > 0$,
		we have $\beta D \act{t} \beta D'$.
	\end{quote}
	
	For some $\beta$, we say that all the appearing fractions in a run are at most $\beta$ if all the 
	fractions appearing in the firing sequence of that run are at most $\beta$.
	Note that repeated applications of the Scaling property imply that 
	if we have a run from $D$ to $D'$ in $\mach$ of the form $D_0 = D \act{} D_1 \act{} \dots \act{} D' = D_k$,
	then by choosing $\beta$ to be so small that $\beta \cdot D_i(j) \le 1$ for every $0 \le i \le k$ and every counter $j$, we can ensure that 
	there is a 1-bounded run from $\beta D$ to $\beta D'$ in which all the appearing fractions are at most $\beta$. 
	Similarly, the converse can also be shown by repeatedly applying the Scaling property.
	This means that it is undecidable to check given a triple $(\mach,C,C')$, whether there exists some $\beta \in (0,1]$ such that $\beta C$ can reach (resp. cover) $\beta C'$ by a 1-bounded run in which all the appearing fractions are at most $\beta$. 
	
	Now, to $\mach$, add two more counters $z$ and $\tilde{z}$. Replace each transition $t = (p,\Delta,q)$
	of $\mach$ by two transitions - the first transition goes from $p$ to a fresh state $s_t$, increases $z$ by 1, decreases $\tilde{z}$ by 1 and does exactly what $t$ does to the remaining counters; the second transition goes from $s_t$ to $q$, decreases $z$ by 1 and increases $\tilde{z}$ by 1.
	Intuitively, the value in $z$ at any point controls the fraction with which $t$ can be fired;
	if the value in $z$ is $\beta$ at some point, then $t$ can only be fired with fraction at most $\beta$.
	Call this new machine $\overline{\mach}$. For any configuration $D$ of $\mach$ and any $\alpha$, let
	$\overline{D}_\alpha$ be a configuration of $\overline{\mach}$ such that $\overline{D}_\alpha(x) = D(x)$ for 
	every counter $x$ of $\mach$, $\overline{D}_\alpha(z) = \alpha$ and $\overline{D}_\alpha(\overline{z}) = 0$. The following claim is easy to prove by induction:
	
	\begin{quote}
		There exists $\beta \in (0,1]$ such that $\beta C$ can reach (resp. cover) $\beta C'$ by an 1-bounded run in which all appearing fractions are at most $\beta$ in $\mach$ iff there exists
		$\beta \in (0,1]$ such that $\overline{\beta C}_\beta$ can reach (resp. cover) $\overline{\beta C'}_\beta$ by an 1-bounded run in
		in $\overline{\mach}$.
	\end{quote}
	
	Now, let the states of the configurations $C, C'$ be $p$ and $q$ respectively.
	To $\overline{\mach}$, we add two new states $p_i$ and $q_f$
	and a fresh counter $st$. From $p_i$ we add a transition to $p$
	which increments each counter $x$ of $\mach$ by $C(x)$, counter $z$ by 1 and decrements counter $st$ by 1.
	Similarly, from $q$ we add a transition to $q_f$ which decrements each counter $x$ of $\mach$
	by $C'(x)$, counter $z$ by 1 and increments counter $st$ by 1. Call this new machine $\mach_1$ and let $C_1$ and $C_1'$ be the configurations of $\mach_1$ whose counter values are all 0, except in the counter $st$ where they have the value 1 and whose states are $q_i$ and $q_f$ respectively. 
	The idea is that the transition from $p_i$ to $p$ non-deterministically guesses a fraction $\beta$,
	sets every counter $x$ of $\mach$ to $\beta C(x)$, sets counter $z$ to $\beta$ and sets $st$ to $1-\beta$.
	From there it simply mimics $\overline{\mach}$ until it reaches $q$. Finally from $q$, it can fire
	the transition to $q_f$, which decrements every counter $x$ of $\mach$ by $\beta C'(x)$, 
	decrements $z$ by $\beta$ and increments $st$ by $\beta$. If this final transition chooses
	to be fired with some fraction other than $\beta$, then either the value of $st$ will be less than 1
	or the value of $z$ will becomes less than 0. With this in mind, the following is then easy to observe:
	
	\begin{quote}
		There exists $\beta \in (0,1]$ such that $\overline{\beta C}_\beta$ can reach (resp. cover)
		$\overline{\beta C'}_\beta$ by an 1-bounded run in $\overline{\mach}$ if and only if $C_1$ can reach (resp. cover)
		$C_1'$ by an 1-bounded run in $\mach_1$.
	\end{quote}
	This completes the reduction and proves the lemma.
\end{proof}

\subsection{Proof of Stage 2: Reduction from resets to zero-row/zero-column matrices.}\label{subsec:appendix-stage2} 

We now show that affine classes that contain a matrix with a zero-row/column can simulate resets.
To this end, let $\cC$ be a non-negative class such that $\cC$ contains a matrix $\bA$ (of size $k \times k$ for some $k$) which either has a 
zero-row or a zero-column. Let $j$ be the row or column of $\bA$ which is completely zero.
Note that 
\begin{multline}\label{eq:two-app}
	\text{If $j$ is a row then for any $\bu \in \qnz^k$, }\bA \cdot \bu  = \bv \\ \text{ where } \bv \in \qnz^k \text{ satisfies } \bv(j) = 0
\end{multline}
and
\begin{equation}\label{eq:three-app}
	\text{If $j$ is a column then for any $\lambda \in \qnz$, } \bA \cdot \lambda \bunit_j = \bzero
\end{equation}

Equation~\ref{eq:two-app} tells us that we can use the matrix to always reset the counter $j$.
Equation~\ref{eq:three-app} tells us that as long as all the counters apart from $j$
on which we are applying $\bA$ have the value 0, then we can reset the counter $j$.
We will exploit these equations now to 
give a reduction from continuous VASS with resets to $\cC$-continuous VASS. 

Let $\mach = (Q,T)$ be a $d$-continuous VASS with resets for some $d$.
For the purposes of the reduction, we can assume that in each transition, $\mach$
resets at most one counter and if a transition does indeed reset some counter,
then the additive update of that transition is $\bzero$.
We can assume this since the undecidability proof given in the previous stage
already holds for continuous VASS with resets with this restriction.
Hence, each transition $t \in T$ is either of the form 
$t = (p,\bI_d,\bb,q)$ or $t = (p,\mathbf{R}^{(d,i)},\bzero,q)$.

From $\mach$, we will create a $n = dk$ dimensional $\cC$-continuous VASS $\mach'$. 
As in the reduction of Theorem~\ref{thm:undec-negative-entries}, the 
counters $x_1 = j, x_2 = j+k, \dots, x_d = j+(d-1)k$ of $\mach'$ will be called the primary counters and
will simulate the counters $1,2,\dots,d$ of $\mach$ respectively. 
The remaining counters of $\mach'$ will be called dummy counters and they will either always have the 
value 0 (if $\bA$ is a zero-column matrix) or their values will in no way affect the values of the primary 
counters (if $\bA$ is a zero-row matrix).

Given a vector $\bu \in \qn^{d}$, we define ext($\bu) \in \qn^{n}$ as the vector which is 0 everywhere, except in counters $x_1, x_2, \dots, x_d$ where the values are respectively
$\bu(1),\bu(2),\dots,\bu(d)$. Furthermore, 

\begin{itemize}
	\item If the $j^{th}$ row of $\bA$ is zero, then we set $S(\bu) = \{\bv \in \qnz^{n} : \bv(x_i) = \bu(i) \text{ for all primary counters } x_i \}$. 
	\item If the $j^{th}$ column of $\bA$ is zero, then we set $S(\bu) = \{\text{ext}(\bu)\}$ 
\end{itemize}

Intuitively, if the $j^{th}$ row is zero, then $S(\bu)$ contains any vector in $\qnz^n$
whose primary counter values coincide with the counter values of $\bu$. On the other hand,
if the $j^{th}$ column is zero, then $S(\bu)$ contains the unique vector in $\qnz^n$
whose primary counter values coincide with the values of $\bu$ and whose dummy values are 0.

The set of states of $\mach'$ will be the same as $Q$. Each transition $t \in T$ of $\mach$
will have a corresponding transition $t'$ in $\mach'$. Before we state this transition $t'$,
we will state two properties that will be satisfied by $t'$. 

\begin{quote}
	\textsc{Property 1: } If $p(\bu') \act{\alpha t'} q(\bv')$ is a step in $\mach'$,
	where $\bu' \in S(\bu)$ for some $\bu$ then there exists $\bv$ such that $\bv' \in S(\bv)$.
\end{quote}

\begin{quote}
	\textsc{Property 2: }  For any configurations $p(\bu),q(\bv)$, we have that  
	$p(\bu) \act{\alpha t} q(\bv)$ is a step in $\mach$ if and only if for every $\bu' \in S(\bu)$,
	there exists $\bv' \in S(\bv)$ such that $p(\bu') \act{\alpha t} q(\bv')$.
\end{quote}

We now proceed to define the transition $t'$ in $\mach'$ that uniquely corresponds to the transition
$t$ in $\mach$. Suppose $t = (p,\bI_d,\bb,q)$ for some states $p,q$ and some vector $\bb$. 
Then we set $t' = (p,\bI_n,\text{ext}(\bb),q)$. It is easily seen that $t'$ satisfies both the properties.

Suppose $t = (p,\mathbf{R}^{(d,i)},\bzero,q)$ for some states $p,q$ and some counter $i$. 
Then we set $t' = (p,\apply{n}{\bA}{(i-1)k+1}),\bzero,q)$. 
Intuitively we are applying $\bA$ to the counters $(i-1)k+1,\dots,(i-1)k+j-1,(i-1)k+j = x_i, (i-1)k+j+1,
\dots, ik$. 
We now show that $t'$ satisfies Properties 1 and 2.
First note that if $\bA$ is a zero-row matrix, then any vector $\bw' \in \qnz^n$ is in $S(\bw)$
where $\bw$ is obtained from $\bw'$ by projecting it onto its primary counters. Hence Property 1
is trivially true for $t'$ when $\bA$ is a zero-row matrix. On the other hand, suppose $\bA$
is a zero-column matrix. In this case, $S(\bw) = \{\text{ext}(\bw)\}$. 
Hence, to prove Property 1, it simply suffices to show that if the value of the dummy counters was 0 before applying $t'$,
then they will remain 0 even after applying $t'$.
By definition of the Application operation, $\apply{n}{\bA}{(i-1)k+1}$ cannot change the value of any counters
in the range $\{1,2,\dots,(i-1)k\} \cup \{ik+1,\cdots,dk\}$. Furthermore, by equation~\ref{eq:three-app} it follows that if all the dummy counters in the range $\{(i-1)k+1,\dots,ik\}$ have value 0 before multiplying by $\bA$,
then they will also have the value 0 afterwards. Hence, this proves Property 1 for $t'$.

Let us now prove Property 2, first in the case when $\bA$ is a zero-row matrix.
Suppose $p(\bu) \act{\alpha t} q(\bv)$ is a step in $\mach$. 
Then $\bv$ is the same as $\bu$, except that its value in the counter $i$ is 0. 
Let $\bu' \in S(\bu)$. The definition of the Application operation
along with equation~\ref{eq:two-app} then means that from $p(\bu')$ we can fire the transition $t'$ with fraction $\alpha$ to reach a configuration $q(\bv')$ where $\bv' \in \qnz^n$ is the same as $\bu'$
in all the counters in the range $\{1,2,\dots,(i-1)k\} \cup \{ik+1,\dots,dk\}$ and furthermore
satisfies $\bv'(x_i) = 0$. Hence, $\bv' \in S(\bv)$. 
Similarly, suppose for every $\bu' \in S(\bu)$, there exists $\bv' \in S(\bv)$ such that
$p(\bu') \act{\alpha t'} q(\bv')$. Take $\bu' = \text{ext}(\bu)$ and let $p(\bu') \act{\alpha t'} q(\bv')$. By definition of the Application operation along with equation~\ref{eq:two-app}, this means
that the primary counters of $\bv'$ have the same values as that of the primary counters of $\bu'$,
except that $\bv'(x_i) = 0$. This means that the counter values of $\bu$ and $\bv$ are the same,
except that there is a possible change at counter $i$, where $\bv(i) = 0$.
Hence, we have that $p(\bu) \act{\alpha t} q(\bv)$. This proves Property 2 in the case
when $\bA$ is a zero-row matrix. The case when $\bA$ is a zero-column matrix is very similar
and uses Equation~\ref{eq:three-app} instead of~\ref{eq:two-app}.

By using induction along with Properties 1 and 2, it is easy to see that
a configuration $p(\bu)$ can cover a configuration $q(\bv)$ in $\mach$
if and only if the configuration $p(\text{ext}(\bu))$ can cover $q(\text{ext}(\bv))$ in $\mach'$.
This completes the proof of correctness of the reduction and proves Theorem~\ref{thm:undec-zero-row-column}.

\section{Proof of Theorem~\ref{thm:undec-weight-overlap-edges}}\label{sec:appendix-weight-overlap-edges}

\undecweight*

Let $\cC$ be a non-negative class. We can assume that $\cC$ does not have any matrices with zero-rows/columns,
as otherwise Theorem~\ref{thm:undec-zero-row-column} applies.
Let $\bA$ be a matrix (of size $k \times k$ for some $k$) which either has a weighted edge
or two overlapping edges. This implies that there is an index $z$ such that 
$\sum_{1 \le i \le k} \bA(i,z) > 1$. This then allows us to deduce 
the following equation. 

\begin{equation}\label{eq:four-app}
	\bA \cdot \bu = \begin{cases}
		\bzero & \text{if } \bu = \bzero\\
		\bv & \text{with } \sum_{1 \le i \le k} \bv(i) > \sum_{1 \le i \le k} \bu(i) \text{ if } \bu(z) > 0
	\end{cases}	
\end{equation}

Indeed, the first part of the equation is clear. For the second part, first note that
since $\bA$ is not a zero-column matrix, for each index $i$, there is some index $\sigma(i)$ such that
$\bA(\sigma(i),i) \ge 1$. Hence, if $\bv = \bA \bu$ then
\begin{align*}
	\sum_{1 \le i \le k} \bv(i) &= \sum_{1 \le i \le k} \sum_{1 \le j \le k} \bA(i,j) \bu(j) \\
	&= \sum_{1 \le j \le k} \sum_{1 \le i \le k} \bA(i,j) \bu(j) \\
	& \ge \sum_{1 \le j \le k, j \neq z} \bA(\sigma(j),j) \bu(j) + \sum_{1 \le i \le k} \bA(i,z) \bu(z)\\
	&> \sum_{1 \le j \le k} \bu(j)	\quad \text{ if } \bu(z) > 0
\end{align*}

Also note that the same argument tells us
\begin{equation}\label{eq:four-app-extra}
	\bA \cdot \bu = \bv \quad \text{with} \sum_{1 \le i \le k} \bv(i) \ge \sum_{1 \le i \le k} \bu(i)
\end{equation}

Let us now prove the undecidability result by giving a reduction from the 1-bounded reachability problem
for continuous VASS with zero-tests. 
To this end, let $(\mach,C,C')$ be an instance of the 1-bounded reachability problem for continuous VASS with zero-tests
where $\mach$ is $d$-dimensional for some $d$. 	
From $\mach$, we will create a $n = d(k+1)$-$\cC$-continuous VASS $\mach'$.
Of the $d(k+1)$ counters of $\mach'$, the counters $x_1 = z, x_2 = z+(k+1), \dots, x_d = z+(d-1)(k+1)$ will be called the primary counters of $\mach'$ and they
simulate the counters $1, 2, \dots, d$ of $\mach$ respectively.
The counters $\overline{x_1} = k+1,\overline{x_2} = 2(k+1),\dots,\overline{x_d} = d(k+1)$ will be called the \emph{complementary counters}
of $x_1, x_2, \dots, x_d$ respectively. Intuitively, each complementary counter $\overline{x_i}$ will
hold 1 minus the value of the counter $x_i$ at any point in a run 
from the initial configuration to the final configuration.
Finally, the remaining counters of $\mach'$
will be called \emph{dummy counters} which are always supposed to have the value 0 in a run
from the initial configuration to the final configuration. 

Given a vector $\delta \in \qn^d$, let ext$(\delta) \in \qn^n$ be the vector which is 0 everywhere,
except in co-ordinates $x_1, x_2, \dots, x_d$ where the values are
$\delta(1),\delta(2),\dots,\delta(d)$ respectively and in the co-ordinates $\overline{x_1},\overline{x_2},\dots,\overline{x_d}$ 
where the values are $-\delta(1),-\delta(2),\dots,-\delta(d)$ respectively. 
Further, given a vector $\bu \in \qn^d$, let $S(\bu) \in \qn^n$ be the vector
which is 0 everywhere except in co-ordinates $x_1,x_2,\dots,x_d$ where it has the values
$\bu(1),\dots,\bu(d)$ and in co-ordinates $\overline{x_1},\overline{x_2},\dots,\overline{x_d}$
where it has the values $1-\bu(1),1-\bu(2),\dots,1-\bu(d)$.
With this notation set up, we are ready to state
the desired reduction. 

The set of states of $\mach'$ will be the same as $Q$.  
We say that a configuration $p(\bu)$ of $\mach'$ is good if 
the values of the dummy counters of $\bu$ are all zero and
for each primary counter $x_i$, $\bu(x_i) + \bu(\overline{x_i}) = 1$.
Note that $p(\bu)$ is good if and only if $\bu = S(\bv)$ for some $\bv \in \qnz^d$.
On the other hand, $p(\bu)$ will be called bad if for some $i$,
the sum of the values of the counters in the range $\{(i-1)(k+1)+1,\dots,i(k+1)\}$
is strictly bigger than 1. Finally an undesirable configuration is one 
which is neither good nor bad.
Our transitions will ensure that starting from a good configuration,
no undesirable configuration can be reached. 
Further, they will also ensure that 
starting from a good configuration
we can reach another good configuration by a run only if 
the run correctly simulated the machine $\mach$. This will then
allow us to reduce the 1-bounded reachability problem for $\mach$
to the reachability problem for $\mach'$.

We will now state the transitions of $\mach'$.
For each $t \in T \cup T_{=0}$ of $\mach$, 
$\mach'$ will have a corresponding transition $t'$.
Suppose $t = (p,\Delta,q) \in T$ is a transition of $\mach$.
Corresponding to $t$, we have the transition $t' = (p,\bI_n,\text{ext}(\Delta),q)$ in $\mach'$.
Intuitively, $t'$ updates the primary counters exactly in the way $t$ updates its counters,
and it updates the complementary counters so that the sum of each primary counter and its 
complementary counter remain the same and does not update the dummy counters at all.
Note that $t'$ satisfies the following properties, 	whose proofs follow immediately from the definition of $t'$ and good and bad configurations.

\begin{quote}
	\textsc{Property 1a): } Suppose $C \act{\alpha t'} D$ is a step in $\mach'$.
	Then $D$ is good (resp. bad) iff $C$ is good (resp. bad).
	
	\textsc{Property 1b): } Suppose $p(\bu)$ and $q(\bv)$ are configurations of $\mach$
	such that $0 \le \bu(i), \bv(i) \le 1$ for every counter $i$.
	Then $p(\bu) \act{\alpha t} q(\bv)$ is a step in $\mach$
	if and only if $p(S(\bu)) \act{\alpha t'} q(S(\bv))$ is a step in $\mach'$.
\end{quote}

Now, suppose $t = (p,i,q) \in \delta_{=0}$ is a zero-test of $\mach$.
Corresponding to $t$, we have the transition $t' = (p,\apply{n}{\bA}{(i-1)(k+1)+1},\bzero,q)$ in $\mach'$.
Intuitively, $t'$ does not update any counter apart from the ones in the range $\{(i-1)(k+1)+1,\dots,i(k+1)-1\}$.
Within this range, it updates the counters by multiplying them with the matrix $\bA$.
Note that by equations~\ref{eq:four-app} and~\ref{eq:four-app-extra}, $t'$ satisfies the following invariants.

\begin{quote}
	\textsc{Property 2a : } Suppose $p(\bu) \act{\alpha t'} q(\bv)$ is a step in $\mach'$.
	If $p(\bu)$ is good and $\bu(x_i) = 0$ then $q(\bv)$ is good.
	If $p(\bu)$ is good and $\bu(x_i) > 0$ or $p(\bu)$ is bad, then $q(\bv)$ is bad.
	
	\textsc{Property 2b: } Suppose $p(\bu)$ and $q(\bv)$ are configurations of $\mach$
	such that $0 \le \bu(i), \bv(i) \le 1$ for every co-ordinate $i$.
	Then $p(\bu) \act{t} q(\bv)$ is a step in $\mach$
	if and only if $p(S(\bu)) \act{\alpha t'} q(S(\bv))$ is a step in $\mach'$ for any $\alpha'$.
\end{quote}

This finishes the construction of $\mach'$. By Properties 1b and 2b it follows that
there is a 1-bounded run between $C$ and $C'$ in $\mach$ if and only if there is a run between
$S(C)$ and $S(C')$ in which all the configurations are good. By Properties 1a and 1b, it 
follows that no run from $S(C)$ to $S(C')$ can have a configuration which is bad or undesirable.
This then proves that $C$ can reach $C'$ by a 1-bounded run in $\mach$ if and only if $S(C)$ can reach $S(C')$ in $\mach'$,
which completes the proof of the theorem.

\section{Proofs of Subsection~\ref{subsec:dec-state-reach}}\label{sec:appendix-dec-state-reach}

First we show the following Soundness property.

\begin{quote}
	\textsc{Soundness: } Suppose $p(\bu) \act{\alpha t} q(\bv)$ is a step in $\mach$.
	Then $(p,\supp{\bu}) \act{t} (q,\supp{\bv})$ is an edge in the support abstraction.
\end{quote}

Indeed, suppose $p(\bu) \act{\alpha t} q(\bv)$ is a step in $\mach$ with $t = (p,\bA,\bb,q)$. Then $\bv = \bA \bu + \alpha \bb$. 
If $x$ is additively decremented by $t$ this means that $\alpha \bb(x) < 0$, and so
it must be the case that, $(\bA \bu)(x) > 0$. Since $\bA$ is a non-negative matrix, the latter condition is equivalent to $\supp{t}^-_x \cap \supp{\bu} \neq \emptyset$. 
Also, if $x$ is additively incremented by $t$ this means that $x \in \supp{\bv}$.
Finally, if $\bv(x) > 0$ then either $(\bA \bu)(x) > 0$ or $\bb(x) > 0$, 
which is equivalent to saying that if $x \in \supp{\bv}$ then either $\supp{t}^-_x \cap S \neq 
\emptyset$ or $x$ is additively incremented by $t$. By the construction
of the support abstraction graph, it follows that the Soundness property is true.

Now we prove the following Completeness property.

\begin{quote}
	\textsc{Completeness: } Suppose $(p,S) \act{t} (q,S')$ is an edge in the support abstraction 
	$SU_\mach$
	and suppose $p(\bu)$ is a configuration of $\mach$ such that $S \subseteq \supp{\bu}$. Then, 
	there exists a configuration $q(\bv)$ and a fraction $\alpha$ such that $S' \subseteq \supp{\bv}$ and
	$p(\bv) \act{\alpha t} q(\bv)$ is a step in $\mach$.
\end{quote}

To this end, let $t = (p,\bA,\bb,q)$ and 
let $\alpha$ be a fraction in $(0,1]$ such that $\alpha < (\bA \bu)(x)/|\bb(x)|$ for every counter
$x$ that is additively decremented by $t$, i.e., whenever $\bb(x) < 0$. 

Note that whenever $\bb(x) < 0$, by definition of edges in $SU_\mach$, 
there must be $y$ such that $y \in S \subseteq \supp{\bu}$ and $\bA(x,y) > 0$. Since, $\bA$ is a non-negative matrix, this means that $(\bA \bu)(x) > 0$ and so it is always possible to choose such a fraction $\alpha$.

Let $\bv = \bA  \bu + \alpha \bb$. We claim that $\bv$ is indeed a non-negative vector.
For the sake of contradiction, suppose $\bv(x) < 0$ for some $x$. This must mean that $(\bA \bu)(x) < -\alpha \bb(x)$.
Since $\bA$ is a non-negative matrix, $(\bA \bu)(x) \ge 0$ and so $\bb(x)$ must be negative.
Hence $-\alpha \cdot \bb(x) = \alpha \cdot |\bb(x)|$ and we initially chose $\alpha$ so that it satisfies
$\alpha \cdot |\bb(x)| < (\bA \bu)(x)$, which leads to a contradiction. 

Now, we claim that $S' \subseteq \supp{\bv}$. Indeed, by the construction of $SU_\mach$,
if $x \in S'$ then either $x$ is additively incremented by $t$ or $\supp{t}_x^- \cap S \neq \emptyset$.
If $x$ is additively incremented by $t$, then $\alpha \bb(x) > 0$ and since $\bA$ is a non-negative matrix,
it follows that $\bv(x) > 0$ as well. On the other hand if $\supp{t}_x^- \cap S \neq \emptyset$,
then there must be $y \in S \subseteq \supp{\bu}$ such that $\bA(x,y) > 0$ and so $(\bA\bu)(x) > 0$.
If $x$ is additively decremented by $t$, then by the choice of $\alpha$, it follows that
$\bv(x) > 0$. If $x$ is not additively decremented by $t$, then $\bv(x) \ge (\bA\bu)(x) > 0$.
Hence, in either case, $x \in \supp{\bv}$ and so $S' \subseteq \supp{\bv}$.

It is then easy to see that $p(\bu) \act{\alpha t} q(\bv)$ is indeed a step in $\mach$ which proves
the Completeness property.

Finally, by using the Soundness and Completeness properties and by induction on the length of the run,
it is then easy to show the required claim.

\begin{quote}
	\textsc{Claim: } $p(\bu)$ can reach a configuration whose state is $q$ in $\mach$ if and only if $(p,\supp{\bu})$
	can reach some vertex of the form $(q,S)$ in the graph $SU_\mach$.
\end{quote}

This concludes the required proof.

\section{Proofs of Subsection~\ref{subsec:dec-cov}}\label{sec:appendix-dec-cov}

\subsection{Proof of Lemma~\ref{lem:basic-facts}}\label{subsec:appendix-basic-facts}

\basicfacts*

\begin{proof}
	All of these claims are clear for the case when $\pi$ is the empty sequence. 
	We will prove all of these claims when $\pi$ is a single step; the general case follows by a straightforward induction on the
	number of steps of $\pi$.
	
	Throughout we fix two configurations $p(\bu), q(\bv)$ and a pair $\gamma t$ where
	$\gamma \in (0,1]$ and $t = (p,\bA,\bb,q)$ is a transition such that $p(\bu) \act{\gamma t} q(\bv)$.
	Hence, this means that $\bv = \bA  \bu + \gamma \bv$.
	
	\paragraph*{Proof of~\ref{item:a}: } 
	Notice that $\alpha \bv = \alpha \bA  \bu + \alpha \gamma \bb = \bA (\alpha \bu) + (\alpha \gamma) \bb$
	and hence $p(\alpha \bu) \act{\alpha \gamma t} q(\alpha \bv)$.
	
	\paragraph*{Proof of~\ref{item:b}: } Notice that
	$\bA  (\bu + \bw) + \gamma \bb = \bA \bu + \gamma \bb + \bA \bw = \bv + \bA \bw$.
	Now, since $\bA$ is a non-negative self-loop matrix, $\bA(x,x) \ge 1$ for every counter $x$ and so $\bA \bw(x) \ge \bw(x)$.
	It then follows that if we let $\bv' = \bA (\bu + \bw) + \gamma \bb$, then $\bv' \ge \bv + \bw$.
	It is then easy to see that $p(\bu + \bw) \act{\gamma t} q(\bv')$.
	
	\paragraph*{Proof of~\ref{item:c}: } Notice that by~(\ref{item:a}), we have
	$p(\alpha \bu) \act{\alpha \gamma t} q(\alpha \bv)$. By~(\ref{item:b}), we have
	$p(\alpha \bu + (1-\alpha) \bu) = p(\bu) \act{\alpha \gamma t} q(\bv')$ where $\bv' \ge \alpha \bv + (1-\alpha) \bu$.
\end{proof}

\subsection{Lemma~\ref{lem:cycle-cov} implies Theorem~\ref{thm:dec-cov-self-loop}}\label{subsec:appendix-cyclic-to-general}

In this subsection, we will prove that Lemma~\ref{lem:cycle-cov} implies
Theorem~\ref{thm:dec-cov-self-loop}.

We claim that a configuration $p(\bu)$ can cover a configuration $q(\bv)$ if and only if 
there is a \emph{witness}, i.e.,
a sequence of configurations $$p_1(\bu_1),p_1(\bu_1'),p_2(\bu_2),p_2(\bu_2'),
\dots,p_k(\bu_k),p_k(\bu_k')$$ such that each $p_i \neq p_j$ for $i \neq j$,
$p_1(\bu_1) = p(\bu), p_k = q, \bu_k' \ge \bv$ and for each $i$, $p_i(\bu_i)$ can cover $p_i(\bu_i')$
and $p_i(\bu_i') \act{} p_{i+1}(\bu_{i+1})$. Indeed, if such a sequence exists,
then by applying Lemma~\ref{lem:basic-facts}~(\ref{item:b}),
each $p_i(\bu_i)$ can cover $p_{i+1}(\bu_{i+1})$. Repeatedly applying Lemma~\ref{lem:basic-facts}~(\ref{item:b}) then implies that $p(\bu) = p_1(\bu_1)$ can cover
$p_k(\bu_k')$. Since $\bu_k' \ge \bv$,
this proves that $p(\bu)$ can cover $q(\bv)$.

On the other hand, suppose we have a run $\rho$
from $p(\bu)$ which covers $q(\bv)$. We proceed by an induction on the length of the run $\rho$,
where the induction hypothesis asserts the existence of a witness for $p(\bu)$ covering $q(\bv)$
such that all the states of the witness appear in the run $\rho$.
The base case of a single step is clear. For the induction step, 
let $i$ be the last position in $\rho$ in which a configuration with state $p$
appears. Let $p(\bw)$ and $p'(\bw')$ be the configuration in positions $i$ and $i+1$ of $\rho$.
We apply the induction hypothesis to the run from $p'(\bw')$ to obtain
a witness $p_1(\bu_1),\dots,p_k(\bu_k)$ for $p'(\bw')$ covering $q(\bv)$.
It then follows that $p(\bu),p(\bw),p_1(\bu_1),\dots,p_k(\bu_k)$
is a witness for $p(\bu)$ covering $q(\bv)$.

Note that from the definition of the step relation, given two states $p,q$, 
we can easily come up with a formula $\phi_{p,q}^{s}$ in ELRA
such that $\phi_{p,q}^s(\bu,\bv)$ is true if and only if $p(\bu) \act{} q(\bv)$, i.e.,
$p(\bu)$ can reach $q(\bv)$ by a single step.
Hence, to decide coverability, we only have to find a witness.
Note that if $p_1(\bu_1),\dots,p_k(\bu_k')$ is a witness then 
the condition that $p_i \neq p_j$ for $i \neq j$ in a witness
implies that $k \le |Q|$. Hence, to find a witness, 
we just have to guess a sequence of states $p_1,\dots,p_k$ 
of the desired kind (which can be done in polynomial time, thanks to the fact that $k \le |Q|$)
and then check for satisfiability of the formula 
\begin{multline*}
	\exists \bu_1,\bu_1',\dots,\bu_k,\bu_k' \bigwedge_{1 \le i \le k} \phi_{p_i}(\bu_i,\bu_i') \land \bigwedge_{1 \le i \le k-1} \phi_{p_i,p_{i+1}}^s(\bu_i',\bu_{i+1}) \\ \land 
	\bu_1 = \bu \land \bu_k' \ge \bv
\end{multline*}

Since satisfiability of formulas in ELRA can be done in \NP, 
this then proves Theorem~\ref{thm:dec-cov-self-loop}.

\subsection{Proof of Proposition~\ref{prop:support-prop}}\label{subsec:appendix-support-prop}

\supportprop*

\begin{proof}
	Suppose $\pi = \alpha_1 t_1, \alpha_2 t_2, \dots, \alpha_\ell t_\ell$.
	Let $p_0(\bu_0) = p(\bu) \act{\alpha_1 t_1} p_1(\bu_1) \act{\alpha_2 t_2} p_2(\bu_2) \dots
	\act{\alpha_\ell t_\ell} p_\ell(\bu_\ell) = p(\bv)$. 	
	
	We let $\rephalf(\pi) = \alpha_1/2 t_1, \alpha_2/4 t_2, \dots, \alpha_\ell/2^\ell t_\ell$ and
	we claim that this satisfies all the claims of the proposition.
	We now show by induction on $k \in [0,\ell]$, that
	\begin{equation*}
		p_0(\bu_0) \act{\rephalf(\pi)[1\cdots k]} p_k(\bw_k) \text{ where } \bw_k \ge \frac{1}{2^k}\bu_k + \sum_{i=0}^{k-1} \frac{1}{2^{i+1}} \bu_i 
	\end{equation*}
	
	Indeed, the base case of $k = 0$ is immediate. For the induction step, note that
	since $p_{k-1}(\bw_{k-1}) \act{\alpha_k t_k} p_k(\bw_k)$, by Lemma~\ref{lem:basic-facts}~(\ref{item:a}),
	it follows that $p_{k-1}\left(\frac{1}{2^k} \bu_{k-1}\right) \act{\frac{\alpha_k}{2^k} t_k}
	p_k\left(\frac{1}{2^k}\bu_k\right)$. 
	Then, if we add $\sum_{i=0}^{k-1} \frac{1}{2^{i+1}} \bu_i$ to both sides,
	by Lemma~\ref{lem:basic-facts}~(\ref{item:b}), we get
	\begin{multline}\label{eq:induction-step}
		p_{k-1}\left(\frac{1}{2^{k-1}} \bu_{k-1} + \sum_{i=0}^{k-2} \frac{1}{2^{i+1}} \bu_i \right) \act{\frac{\alpha_k}{2^k} t_k} p_k(\bw_k) \\ \text{ where }  \bw_k \ge
		\frac{1}{2^k}\bu_{k} + \sum_{i=0}^{k-1} \frac{1}{2^{i+1}} \bu_i
	\end{multline}
	
	By induction hypothesis, we have that $$p_0(\bu_0) \act{\rephalf(\pi)[1\cdots (k-1)]} p_{k-1}(\bw_{k-1})$$
	where $\bw_{k-1} \ge \frac{1}{2^{k-1}}\bu_{{k-1}} + \sum_{i=0}^{k-2} \frac{1}{2^{i+1}} \bu_i$.
	Hence, by Lemma~\ref{lem:basic-facts}~(\ref{item:b}) and Equation~\ref{eq:induction-step}
	we can conclude the proof of the induction. 
	
	Let $\bw = \bw_\ell$. Suppose a counter $x$ was additively incremented (resp. decremented)
	by some transition $t$ in $\rephalf(\pi)$. Since $\supp{\rephalf(\pi)} = \supp{\pi}$,
	it follows that $t$ is also a transition in $\pi$. Hence,  
	there must be some $k$ such that $\bu_k(x) > 0$ or $\supp{t}^-_x \cap \supp{\bu_k} \neq \emptyset$
	respectively. Since $\bw \ge 1/2^{k+1} \bu_k$,
	it would then follow that a similar property is satisfied by $\bw$ as well.
	
	Suppose a counter $x$ was pumped by $t$ in the run $p(\bu) \act{\pi} p(\bv)$. 
	By definition of pumping, there is some $\bu_k$ such that $\supp{t}^p_x \cap \supp{\bu_k} \neq \emptyset$. Since $\bw \ge 1/2^{k+1} \bu_k$, we have $\supp{t}^p_x \cap \supp{\bw} \neq \emptyset$ and this finishes the proof.
\end{proof}

\subsection{Proof of the Pumping Counters Lemma (Lemma~\ref{lem:pumping-counters})}~\label{subsec:appendix-pumping-counters}

To prove this lemma, we first prove a bunch of small results.

\begin{proposition}[Marking equation]\label{prop:marking-equation}
	Let $p(\bu) \act{\pi} q(\bv)$ be a run such that $\pi = \alpha_1 t_1, \alpha_2 t_2, \dots, \alpha_\ell t_\ell$ and each $t_i = (p_{i-1},\bA_i,\bb_i,p_i)$. Then,
	\begin{equation*}
		\bv = \left(\prod_{j=\ell}^1 \bA_j\right) \bu +  \sum_{j=1}^\ell \left(\prod_{k=\ell}^{j+1} \bA_k \right) \alpha_j \bb_j
	\end{equation*}
	where $\prod_{j=\ell}^1 \bA_j$ denotes the matrix product $\bA_\ell \cdot \bA_{\ell-1} \cdots \bA_1$;
	similarly $\prod_{k=\ell}^{j+1} \bA_k$ denotes the product $\bA_\ell \cdot \bA_{\ell-1} \cdots \bA_{j+1}$.
\end{proposition}

\begin{proof}
	Follows by induction on the length of $\pi$ and the definition of a step.
\end{proof}

\begin{proposition}[Fundamental Property of Self-loop Matrices]\label{prop:fun-prop-self-loop}
	Suppose $\bA_1, \bA_2, \dots, \bA_\ell$ are non-negative self-loop matrices. Then,
	$\bB = \prod_{k=1}^\ell \bA_k$ is also a non-negative self-loop matrix
	such that $\bB(i,j) \ge \bA_k(i,j)$ for every $i,j$ and $k$
\end{proposition}

\begin{proof}
	Follows easily by induction on $\ell$.
\end{proof}

\pumpingcounters*

\begin{proof}
We prove the lemma in two steps.

\paragraph*{Step 1: Fire $\pi/2$ from $p(\bu)$. } 

Let $\pi = \alpha_1 t_1, \alpha_2 t_2, \dots, \alpha_\ell t_\ell$.
Let each $t_i = (p_{i-1},\bA_i,\bb_i,p_i)$. By the marking equation we have that
\begin{equation*}
	\bv = \left(\prod_{j=\ell}^1 \bA_j\right) \bu +  \sum_{j=1}^\ell \left(\prod_{k=\ell}^{j+1} \bA_k \right) \alpha_j \bb_j
\end{equation*}	
Let
\begin{equation*}
	\bz = \sum_{j=1}^\ell \left(\prod_{k=\ell}^{j+1} \bA_k \right) \alpha_j \bb_j
\end{equation*}

By Lemma~\ref{lem:basic-facts}~(\ref{item:b}), it follows that
$p(\bu) \act{\pi/2} p(\bu')$ where $\bu' \ge \bv/2 + \bu/2$. 
This completes Step 1.

\paragraph*{Step 2: Choosing $n$ and firing $1/2n$-scaled down versions of $\pi$. }

Let $m = \min\{\bv(x) : \bv(x) \neq 0\}$. Choose $n$ so that $nm/2> K$. 
By assumption, for every counter $x \in S$ there exists some counter $c_x$ such that
$c_x \in \supp{t}^p_x \cap \supp{\bv}$.
We now construct a sequence of configurations $p(\bu_1),p(\bu_2),\dots,p(\bu_{n+1})$
such that each $p(\bu_i)$ satisfies the following property:
For every counter $x$,
\begin{itemize}
	\item Property 1: If $x \notin S$, then $\bu_i(x) \ge \frac{\bv(x)}{2} + \frac{(n-i+1)\bu(x)}{2n} + \frac{(i-1)\bv(x)}{2n}$
	\item Property 2: If $x \in S$, $\bu_i(x) \ge \frac{\bv(x)}{2} + \frac{(n-i+1)\bu(x)}{2n} + \frac{(i-1)\bv(x)}{2n} + \frac{(i-1)\bv(c_x)}{2}$
\end{itemize}

Notice that if we let $p(\bu_1) = p(\bu')$, then $p(\bu_1)$ satisfies the above property.
Now, assume that $p(\bu_i)$ already satisfies this property. Hence, $\bu_i \ge \bu/2n$. By
Lemma~\ref{lem:basic-facts}~(\ref{item:a} and~\ref{item:b}) it then follows that 
we can fire $\pi/2n$ from $\bu_i$. Let $p(\bu_i) \act{\pi/2n} p(\bu_{i+1})$.
By the marking equation, we have that
\begin{equation}\label{eq:mark}
	\bu_{i+1} = \left(\prod_{j=\ell}^1 \bA_j\right) \bu_i +  \frac{1}{2n} \sum_{j=1}^\ell \left(\prod_{k=\ell}^{j+1} \bA_k \right) \alpha_j \bb_j = \left(\prod_{j=\ell}^1 \bA_j\right) \bu_i +  \frac{\bz}{2n}
\end{equation}

Let $\bB = \left(\prod_{j=\ell}^1 \bA_j\right)$.
By the induction hypothesis, we have that $\bu_i = \frac{\bv}{2} + \frac{(n-i+1)}{2n} + \frac{(i-1)\bv}{2n} + \mathbf{\bu'_i}$ for some non-negative vector $\mathbf{\bu'_i}$
such that $\bu'_i(x) \ge \frac{(i-1)\bv(c_x)}{2}$ for every $x \in S$.
Hence, 
\begin{equation*}
	\bB \bu_i = \bB \left(\frac{\bv}{2} + \frac{(n-i+1)\bu}{2n} + \frac{(i-1)\bv}{2n} + \mathbf{\bu'_i}\right)
\end{equation*}

By the Fundamental property of non-negative self-loop matrices
$\bB$
is also a non-negative self-loop matrix and hence all of its diagonal entries are at least 1.
Hence,
\begin{multline*}
	\bB  \left(\frac{\bv}{2} + \frac{(n-i+1)\bu}{2n} + \frac{(i-1)\bv}{2n} + \mathbf{\bu'_i}\right) \\ \ge \bB  \frac{\bv}{2} +
	\bB \frac{\bu}{2n} + \frac{(n-i)\bu}{2n} + 
	\frac{(i-1)\bv}{2n} + \mathbf{\bu'_i}
\end{multline*}

Combining the last three equations gives us,
\begin{multline*}
	\bu_{i+1}  \ge \bB \frac{\bv}{2} +
	\bB \frac{\bu}{2n} + \frac{(n-i)\bu}{2n} + 
	\frac{(i-1)\bv}{2n} + \mathbf{\bu'_i} + \frac{\bz}{2n}
\end{multline*}

By equation~\ref{eq:mark}, $\bB \frac{\bu}{2n} + \frac{\bz}{2n} = \frac{\bv}{2n}$ and so
\begin{equation*}
	\bu_{i+1} \ge \bB \frac{\bv}{2} + \frac{(n-i)\bu}{2n} + \frac{i\bv}{2n} + \mathbf{\bu'_i}
\end{equation*}

If $x \notin S$, then since $\bB$ is a self-loop matrix, we have that $\left(\bB \frac{\bv}{2}\right)(x) \ge \frac{\bv(x)}{2}$ and so Property 1 for $\bu_{i+1}$ is satisfied.
On the other hand, suppose $x \in S$. Then $x$ is pumped in the run $\rho$
and so this means that for some $i$, $\bA_{i}(x,c_x) > 0$ if $c_x \neq x$ and $\bA_{i}(x,c_x) > 1$
if $c_x = x$. By the Fundamental property of self-loop matrices, it follows
that $\bB(x,c_x) \ge \bA_i(x,c_x)$ and so it follows that $\left(\bB \frac{\bv}{2}\right)(x) \ge  
\frac{\bv(x)}{2}  + \frac{\bv(c_x)}{2}$. Since $\bu'_i(x) \ge \frac{(i-1)\bv(c_x)}{2}$ for every $x \in S$,
it follows that Property 2 is also satisfied for $\bu_{i+1}$.

Now, by assumption on $\bv$, for every $x \in S$, $\bv(c_x) > 0$. By definition 
of $m$ this means that $\bv(c_x) \ge m$.
Hence, when Property 1 and 2 are applied to $p(\bu_{n+1})$, we have
\begin{itemize}
	\item If $x \notin S$, then $\bu_{n+1}(x) \ge \bv(x)$
	\item If $x \in S$, then $\bu_{n+1}(x) \ge \bv(x) + \frac{n\bv(c_x)}{2} \ge nm/2 > K $
\end{itemize}
Hence, if we set $\bw = \bu_{n+1}$, then the lemma becomes true.
\end{proof}

\subsection{Proof of the Coverability Characterization lemma (Lemma~\ref{lem:cov-charac})}\label{subsec:appendix-cov-charac}

\covcharac*

\begin{proof}
	Suppose $p(\bu)$ can cover $p(\bv)$ in $\mach$.
	Hence, $p(\bu) \act{\pi} p(\bw)$ for some $\bw \ge \bv$.
	Let $X$ be the 
	set of counters not pumped by this run. Clearly we can set $\pi_{\textsf{fwd}} = \pi,
	\pi' = \pi$ and conditions 2, 3 and 4 are all satisfied.
	To prove condition 1, let the run $p(\bu) \act{\pi} p(\bw)$. 
	be $p_0(\bu_0) = p(\bu) \act{\alpha_1 t_1} p_1(\bu_1) \dots
	\act{\alpha_\ell t_\ell} p_\ell(\bu_\ell) = p(\bw)$
	where each $t_i = (p_{i-1},\bA_i,\bb_i,p_i)$.
	If $x \in X$, then $x$ is not pumped in this run and so for every $i$, 
	$\bu_i(x) = \bu_{i-1}(x) + \alpha_i \bb_i$.
	It then follows that $p(\bu_X) \act{\pi_X} p(\bw_X)$ and so condition 1 is satisfied.
	

	
	Now, suppose there exists a set of counters $X$ and firing sequences $\pi', \pi_{\textsf{fwd}}$ over $\mach$ which satisfies all the four conditions. Let $\rho$ be the $p$-admissible run 
	obtained from $p(\bu)$ by firing $\pi_{\textsf{fwd}}$.
	Condition 3 tells us that if $x \notin X$, then $x$ is pumped by $\rho$.
	Without loss of generality, we can assume that if $x \in X$, then $x$ is not pumped
	by $\rho$. Indeed, if there is a set of counters $S \subseteq X$,
	that is pumped by $\rho$, then we set $X' = X \setminus S$.
	Since condition 1 guarantees us that $p(\bu_X) \act{\pi_X'} p(\bw_X)$ is a run in $\mach_X$,
	we have a run $p(\bu_{X'}) \act{\pi'_{X'}} p(\bw_{X'})$ in $\mach_{X'}$. 
	This new $X'$ along with $\pi'$ and $\pi_{\textsf{fwd}}$ still
	satisfies conditions 1 - 4. Hence, for the rest of the proof, 
	we will assume that a counter is pumped by $\rho$
	iff it is not in $X$. 
	
	As mentioned before, let $\rho$ be the run obtained from $p(\bu)$ by firing $\pitfwd$.
	We will now prove the required claim in multiple steps.
	
	\paragraph*{Step 1: Constructing an alternate run in $\mach_X$. }
	In this step, we will concentrate on $\mach_X$. We will modify the run $p(\bu_X) \act{\pi_X'} p(\bw_X)$
	in $\mach_X$ into a different run with the same end-points which we will then use in the next step.
	
	Let us look at the firing sequences $\pi_{\text{fwd}_X} = \alpha_1 t_{1_X}, \dots, \alpha_\ell {t_\ell}_X$
	and $\rephalf\left(\pi_{\text{fwd}_X}\right) = \alpha_1/2 t_{1_X}, \dots, \alpha_\ell/2^\ell t_{\ell_X}$ in $\mach_X$.
	Because $\rho$ did not pump any counters in $X$,
	it follows that we can fire $\pi_{\text{fwd}_X}$ from $p(\bu_X)$ in $\mach_X$. Proposition~\ref{prop:support-prop}
	then tells us that we can fire $\rephalf\left(\pi_{\text{fwd}_X}\right)$ from $p(\bu_X)$ to reach a configuration $p(\bz'')$
	such that for every transition $t_X \in \supp{\pi_{\text{fwd}_X}} = \supp{\rephalf\left(\pi_{\text{fwd}_X}\right)}$ and every counter $x$,
	if $x$ is additively incremented or decremented by $t_X$, then $\bz''(x) > 0$. 
	
	The proof of~\cite[Proposition 4.5]{blondinLogicsContinuousReachability2017} shows how to use this vector $\bz''$ 
	along with
	condition 1 to show that we can pick a small enough fraction $\lambda$
	to get a run of the following form in the continuous VASS $\mach_X$.
	\begin{equation}\label{eq:run-in-restricted}
		p(\bu_X) \act{\lambda \rephalf\left(\pi_{\text{fwd}_X}\right)} p(\bz') \act{\eta_X} p(\bw_X)	
	\end{equation}
	for some firing sequence $\eta_X = \beta_1 r_{1_X}, \dots, \beta_k r_{k_X}$ in $\mach_X$
	such that $\supp{\eta_X} = \supp{\pi_X} = \supp{\pi_{\text{fwd}_X}}$.
	
	Note that this induces a firing sequence $\eta = \beta_1 r_1, \dots, \beta_k r_k$
	in $\mach$ where each $r_i = (p_{i-1},\bA_i,\bb_i,p_i)$ and $\supp{\eta} = \supp{\pi} = \supp{\pitfwd}$.
	Define
	\begin{equation}~\label{eq:epsilon-again}
		\epsilon = \max\left\{\left(\sum_{j=1}^i \left(\prod_{e=i}^{j+1} \bA_e \right) \alpha_j |\bb_j|\right)(x) : 1 \le i \le \ell', 1 \le x \le d \right\}	
	\end{equation}
	Intuitively $\epsilon$ is the maximum amount that any step of $\eta$ can decrement on any counter.
	This completes Step 1. In the next steps,
	we will show how to use $\eta$ to construct what we want.
	
	\paragraph*{Step 2: Preparing to fire $\eta$ in $\mach$. }
	In this step, we will concentrate on $\mach$. We will figure out a way to reach
	a configuration from $p(\bu)$ in $\mach$, from which we can fire $\eta$.
	
	Condition 3 tells us that $\pi_{\textsf{fwd}} = \alpha_1 t_1, \dots, \alpha_\ell t_\ell$ can be fired from $p(\bu)$.
	By Proposition~\ref{prop:support-prop}, it follows that there 
	is a firing sequence $\rephalf(\pitfwd) = \alpha_1/2 t_1, \dots, \alpha_\ell/2^\ell t_\ell$ 
	that can be fired from $p(\bu)$
	to reach some $p(\bu'')$ such that if a counter $x$
	is pumped by some transition $t$ in $\rho$, then $\supp{t}^p_x \cap \supp{\bu''} \neq \emptyset$.
	By Lemma~\ref{lem:basic-facts}~(\ref{item:c}), we then have that $p(\bu) \act{\lambda \rephalf(\pitfwd)} p(\bu')$ where
	$\bu' \ge \lambda \bu'' + (1-\lambda) \bu$ and so $\supp{\bu''} \subseteq \supp{\bu'}$. 
	Hence, if a counter $x$
	is pumped by some transition $t$ in $\rho$, then $\supp{t}^p_x \cap \supp{\bu'} \neq \emptyset$. 
	
	We now claim that $\bu'_X \ge \bz'$. Indeed, by definitions of $\bu'$ and $\bz'$ we have
	$$p(\bu) \act{\lambda \rephalf(\pitfwd)} p(\bu') \text{ is a run in } \mach$$
	and
	$$p(\bu_X) \act{\lambda \rephalf\left(\pi_{\text{fwd}_X}\right)} p(\bz') \text{ is a run in } \mach_X$$
	
	Now, the first run is of the form $p(\bu) = q_0(\bu_0) \act{\lambda \alpha_1/2 t_1} q_1(\bu_1) \dots$
	$\act{\lambda \alpha_\ell/2^\ell t_\ell} q_\ell(\bu_\ell) = p(\bu')$ and 
	the second run is of the form $p(\bu_X) = q_0({\bz'_0}) \act{\lambda \alpha_1/2 t_{1_X}} q_1(\bz'_1) \dots
	\act{\lambda \alpha_\ell/2^\ell t_{\ell_X}} q_\ell(\bz'_\ell) = p(\bz')$. 
	By induction on $i$, it is easy to verify that $\bu_{i_X} \ge {\bz'_i}$. Hence, $\bu'_X \ge \bz'$.
	This completes Step 2.
	
	\paragraph*{Step 3: Firing $\eta$ to cover $p(\bw)$. }
	In the third step, we focus on firing $\eta$. This will ultimately allow us to cover $p(\bw)$
	and since $\bw \ge \bv$, this will cover $p(\bv)$.
	
	From the previous step, we know that $p(\bu) \act{\lambda \rephalf(\pitfwd)} p(\bu')$
	such that whenever a counter $x \notin X$ then there is a transition $t \in \supp{\pi} = \supp{\pitfwd} = \supp{\lambda\rephalf(\pitfwd)}$ such that $\supp{t}^p_x \cap \supp{\bu'} \neq \emptyset$. 
	Applying the Pumping Counters lemma, we now get 
	a run $p(\bu) \act{*} p(\mathbf{s})$ such that $\mathbf{s} \ge \bu'$
	and for every counter $x \notin X$, $\mathbf{s}(x) > \ell' \epsilon + \bw(x)$.

	Now, let $p(\bz') = p_0(\bz_0) \act{\beta_1 r_{1_X}} p_1(\bz_1) \act{\beta_2 r_{2_X}} \dots \act{\beta_{k} r_{k_X}} p(\bz_k) = p(\bw_X)$ be the run obtained from firing $\eta_X$ in $\mach_X$ from Equation~\ref{eq:run-in-restricted}. 
	By induction on $i$, we will now show that
	$p(\bs) = p_0(\bs_0) \act{\beta_1 r_1} p_1(\bs_1) \dots \act{\beta_i r_i} p_i(\bs_i)$
	is a run in $\mach$
	where $\bs_i(x) > (\ell' - i) \epsilon + \bw(x)$ if $x \notin X$
	and $\bs_i(x) \ge \bz_i(x)$ if $x \in X$. Indeed, the base case of $i = 0$ follows from the definition of $\bs$ and
	the fact that $\bs_X \ge \bu'_X \ge \bz'$.
	For the induction step, notice that by definition of $\epsilon$, each $r_i$ can decrement any counter by at most $\epsilon$
	and so $\bs_{i+1}(x) \ge \bs_i(x) - \epsilon > (\ell' - (i+1)) \epsilon + \bw(x)$ if $x \notin X$.
	Further, it is easy to see that $\bs_{i+1}(x) \ge \bz_{i+1}(x)$. Hence, 
	it follows that from $p(\bs)$ we can reach a configuration $p(\bs_k)$ such that for every $x \notin X$,
	we have $\bs_{k}(x) > \bw(x)$ and for every $x \in X$ $\bs_{k}(x) \ge \bz_{k}(x) = \bw(x)$. 
	Hence, we have a run from $p(\bz')$ covering $p(\bw)$ which in turn covers $p(\bv)$.
	Hence, $p(\bz')$ can cover $p(\bv)$ and since $p(\bu)$ can reach $p(\bz')$, $p(\bu)$ can also
	cover $p(\bv)$ which finishes the proof of the lemma.	
\end{proof}

\subsection{Proof of Lemma~\ref{lem:far}}\label{subsec:appendix-far}

Before we prove Lemma~\ref{lem:far}, we first prove another lemma.

\begin{lemma}\label{lem:sufficient}
	Let $p(\bu)$ be a configuration, $p$ be some state, $Y$ be some set of counters 
	and let $\rho = t_1,t_2,\dots,t_\ell$ be some
	sequence of transitions. If $\rho$ is a valid path in $\gr_\mach$ from $p$ to $p$
	such that 
	\begin{itemize}
		\item For every counter $y$ and every transition $t_i$, if $\supp{t_i}^-_y$ is non-empty,
		then $$\supp{t_i}^-_y \bigcap \left(\supp{\bu} \cup \bigcup\limits_{1 \le j < i} \supp{t_j}^+\right) \neq \emptyset$$
		\item For every counter $y \in Y$, there exists some transition $t_i$ such that 
		$$\supp{t}^p_y \bigcap \left(\supp{\bu} \cup \bigcup\limits_{1 \le j \le \ell} \supp{t_j}^+\right) \neq \emptyset$$
	\end{itemize}
	then there exists a $p$-admissible sequence $\pi$ from $p(\bu)$ such that
	the resulting run pumps all counters in $Y$.
\end{lemma}

\begin{proof}
	Since $\rho$ is a valid path in $\gr_\mach$, it follows that
	we can let each $t_i = (p_{i-1},\bA_i,\bb_i,p_i)$ where $p_0 = p_\ell = p$.
	By induction on $i$, we will now show that there is a run of the form
	$$p_0(\bu_0) = p(\bu) \act{\alpha_1 t_1, \dots, \alpha_i t_i} p_i(\bu_i)$$
	for some fractions $\alpha_1,\dots,\alpha_i$ and some non-negative vector $\bu_i$
	such that $\supp{\bu} \subseteq \supp{\bu_i}$ and for every counter $x$ and every transition $t_j$ with $j \le i$,
	\begin{itemize}
		\item If $x$ is additively incremented by $t_j$, then $\bu_i(x) > 0$.
		\item If $x$ is additively decremented by $t_j$, then $\supp{t_j}^-_x \cap \supp{\bu_i} \neq \emptyset$.
	\end{itemize}
	The base case of $i = 0$ is clear.
	
	For the induction step, assume that we have a run 
	$p(\bu) \act{\alpha_1 t_1, \dots, \alpha_i t_i} p_i(\bu_i')$
	of the desired form.
	Now for each counter $x$ that is additively decremented by $t_{i+1}$, by assumption, 
	there is some counter $c_x \in \supp{t_{i+1}}^-_x$ such that either $c_x \in \supp{\bu}$
	or $c_x$ is additively incremented by some $t_j$ with $j \le i+1$. In either case
	this implies that $\bu_i(c_x) > 0$. Choose $\beta_i$ to be so small
	that $0 < \beta_i < -\bu_i(c_x)/\bb_{i+1}(x)$ for any $x$ that is additively decremented
	by $t_{i+1}$. It is then easy to see that from $\bu_i$ we can fire $\beta_i t_i$ 
	to arrive at a configuration of the desired form $p_{i+1}(\bu_{i+1})$,
	which then completes the induction.

	Now, we have shown that there is a $p$-admissible path from $p(\bu)$ to some $p(\bw)$
	such that $\supp{\bu} \subseteq \supp{\bw}$ and for every counter $x$ and every transition $t_j$ with $j \le \ell$,
	\begin{itemize}
		\item If $x$ is additively incremented by $t_j$, then $\bw(x) > 0$.
		\item If $x$ is additively decremented by $t_j$, then $\supp{t_j}^-_x \cap \supp{\bw} \neq \emptyset$.
	\end{itemize}
	Let $\epsilon$ be a very small fraction so that $\epsilon < -\bw(x)/(\bb_i(y) \cdot \ell)$ for any
	$x$ such that $\bw(x) > 0$ and any $y, i$ such that $\bb_i(y) < 0$. Let $m = \max \{-\bb_i(y) : \bb_i(y) < 0\}$.
	We now show that from $p(\bw)$ we can fire $\epsilon t_1, \epsilon t_2, \dots, \epsilon t_\ell$
	such that the resulting run pumps all the counters in $Y$.
	
	Indeed, by induction on $i$, we will show that there is a run of the form
	$$p_0(\bw_0) = p(\bw) \act{\epsilon t_1, \dots, \epsilon t_i} p_i(\bw_i)$$
	for some non-negative vector $\bw_i$ such that $\supp{\bw} \subseteq \supp{\bw_i}$ and $\bw_i(x) > (\ell-i) \epsilon m$
	for every $x \in \supp{\bw}$ and for every $y \in S$, if $\supp{t_i}^p_y$ is non-empty, then $y$ 
	was pumped by $t_i$.
	The base case of $i = 0$ is clear.
	
	For the induction step, assume that we have already constructed such a run till $p_i(\bw_i)$.
	Consider $t_{i+1}$. If $t_{i+1}$ additively decrements some counter $x$, then
	by assumption on $\bw$, it follows that $\supp{t_{i+1}}^-_x \cap \supp{\bw} \neq \emptyset$.
	Hence, there is some counter $c_x$ such that $c_x \in \supp{t_{i+1}}^-_x \cap \supp{\bw}$.
	Since $\supp{\bw} \subseteq \supp{\bw_i}$, it follows that $\bw_i(c_x) > (\ell-i) \epsilon m$.
	Notice that $t_{i+1}$ can additively decrement some counter by at most $m$
	and so we can now fire $\epsilon t_{i+1}$ and all counters that were not additively decremented
	would be at least as large as their value in $\bw_i$ and all counters
	that were additively decremented would be strictly larger than $(\ell-i-1) \epsilon m$.
	It then follows that $\supp{\bw} \subseteq \supp{\bw_i} \subseteq \supp{\bw_{i+1}}$
	and $\bw_{i+1}(x) > (\ell-i-1) \epsilon m$ for every $x \in \supp{\bw}$.
	The only thing that we have to show is that if $y \in S$ is such that $\supp{t_{i+1}}^p_y$
	is non-empty, then $y$ was pumped by $t_{i+1}$.
	
	By assumption, if $y \in S$ and $x \in \supp{t_{i+1}}^p_y$ then
	$x$ is either in $\supp{\bu}$ or $x$ is additively incremented by some $t_j$.
	By definition of $\bw$, it follows that $x$ is in $\supp{\bw} \subseteq \supp{\bw_i}$.
	Hence, when firing $\epsilon t_{i+1}$ from $\bw_i$, we must have pumped counter $y$.
	Hence, the induction step is complete and the proof is done.
\end{proof}

\far*

\begin{proof}
	Suppose there is a firing sequence $\pi$ such that $\supp{\pi} = S$ and firing $\pi$ from $p(\bu)$ results
	in a $p$-admissible run which pumps all the counters in $Y$.
	Let $t_1, t_2, \dots, t_{|S|}$ be such that $t_1$ is the first distinct transition
	that appears in $\pi$, $t_2$ is the second distinct transition that appears in $\pi$ 
	and so on.
	Let $f(t_i) = i$. It is easy to see that $f$ satisfies the first three conditions.
	Furthermore, note that if some counter $x$ is additively decremented by some $t_j$,
	then there must be $y \in \supp{t_j}^-_x$ that either had a non-zero value at $\bu$
	or that must have been additively incremented by some transition in the set $\{t_1,\dots,t_{j-1}\}$.
	Indeed, if this is not the case, the value of all counters in $\supp{t_j}^-_x$ will be 
	zero when $t_j$ is fired, which will result in $x$ having a negative value.
	Hence condition 4 is also satisfied. Similarly we can show that condition 5 is also satisfied.
	
	Suppose there exists an injection $f$ that satisfies all the conditions mentioned above.
	Let $t_1$ be the transition which has the least $f$-value among all transitions in $S$,
	i.e., $f(t_1) \le f(t)$ for every $t \in S$ and let $t_2$ be the transition
	which has the second least $f$-value among all transitions in $S$ and so on.
	From the conditions, it follows that there is a path $r_1,r_2,\dots,r_\ell$ 
	in $\gr_\mach$
	from $p$ to $p$ using only the transitions from $S$ such that
	\begin{itemize}
		\item For every counter $y$ and every transition $r_i$, if $\supp{r_i}^-_y$ is non-empty,
		then $$\supp{r_i}^-_y \bigcap \left(\supp{\bu} \cup \bigcup\limits_{1 \le j < i} \supp{r_j}^+\right) \neq \emptyset$$
		\item For every counter $y \in Y$, there exists some transition $r_i$ such that 
		$$\supp{r_i}^p_y \bigcap \left(\supp{\bu} \cup \bigcup\limits_{1 \le j \le \ell} \supp{r_j}^+\right) \neq \emptyset$$
	\end{itemize}
	
	Applying Lemma~\ref{lem:sufficient} then finishes the proof.
\end{proof}

\section{Proofs of Subsection~\ref{subsec:dec-reach-perm}} \label{sec:appendix-dec-reach-perm}

We show that the reachability problem for $\cC$-continuous VASS is in \NEXP \ if 
$\cC$ contains only permutation matrices. We do this by giving an exponential-time reduction to the 
reachability problem for continuous VASS which is in \NP~\cite[Corollary 4.10]{blondinLogicsContinuousReachability2017}.

To this end, let $\cC$ be a non-negative class which contains only permutation matrices.
Let $\mach = (Q,T)$ be a $d$-dimensional $\cC$-continuous VASS for some $d$.
Since all matrices in $\cC$ are non-negative, 
without loss of generality, we can assume that if $(p,\bA,\bb,q) \in T$ then
either $\bb = \bzero$ or $\bA = \bI$. Indeed, if there is some transition $(p,\bA,\bb,q)$ which
does not satisfy this assumption, then we add a new state $q'$ and replace this transition
with $(p,\bA,\bzero,q')$ and $(q',\bI,\bb,q)$. It can be easily verified that this new state 
preserves the reachability, coverability and state-reachability relation of the given affine VASS.

With this assumption, we now construct a $d$-dimensional continuous VASS $\mach' = (Q',T')$
that simulates $\mach$. Intuitively, the idea is that since each affine operation of $\mach$
is only a renaming of the counters, we can store this renaming in the states of the continuous VASS.
With this intuition, we now define $\mach'$.

The set of states $Q'$ of $\mach'$ will be $\{q_\sigma : q \in Q, \sigma \in \mathcal{S}_d\}$.
Intuitively, each state of $\mach'$ stores a state of $\mach$ and remembers the current renaming
that needs to be done to the counters. For each transition $t \in T$ of $\mach$,
and each permutation $\pi$,
there will be a transition $t_\pi \in T'$ of $\mach'$ defined as follows:

\begin{itemize}
	\item If $t = (p,\bP_{\sigma},\bzero,q)$, then $t_\pi$ is defined
	as $t_\pi = (p_\pi, \bzero, q_{\sigma \circ \pi})$. Intuitively, if we are remembering that
	the current renaming of the counters is $\pi$, then applying $\bP_\sigma$ amounts
	to remembering that the new renaming of the counters is $\sigma \circ \pi$.
	\item If $t = (p,\bI,\bb,q)$, then $t_\pi$ is defined
	as $t_\pi = (p_\pi, \bP_{\pi} \bb, q_\pi)$. Intuitively, if we are remembering that
	the current renaming of the counters is $\pi$, then adding $\bb$ in $\mach$
	amounts to adding $\bP_{\pi} \bb$ in $\mach'$.
\end{itemize}

A very easy induction then shows that 
\begin{quote}
	$p(\bu)$ can reach $q(\bv)$ in $\mach$ if and only if for some permutation $\sigma$, 
	$p_{id}(\bu)$ can reach $q_\sigma(\bP_{\sigma} \bv)$ in $\mach'$ where $id$ is the identity permutation.
\end{quote}

Our algorithm for solving reachability in $\mach$ is then as follows: Construct $\mach'$, guess a permutation and check whether $p_{id}(\bu)$ can reach $q_\sigma(\bP_{\sigma} \bv)$ in $\mach'$. Since $\mach'$ can be constructed in exponential time 
and since reachability in continuous VASS is in \NP, it follows that reachability
of $\mach'$ can be decided in \NEXP.

\section{Proofs of Subsection~\ref{subsec:lower-bound-zero-row-column}}\label{sec:appendix-lower-bound-zero-row-column}

Here we will prove that state-reachability for continuous VASS with resets is \PSPACE-hard.
To prove this, we give a reduction from state-reachability for Boolean programs. Intuitively, a Boolean program
is a finite-state automaton which has access to $d$ Boolean variables (for some $d$), which it can either set to 0 or 1 and which can be tested for its current value. 
Formally, a $d$-Boolean program
is a tuple $\mach = (Q,T)$ where $Q$ is a finite set of control states
and $T \subseteq Q \times \{\text{test}_j(i), \text{set}_j(i) : 1 \le i \le d, j \in \{0,1\}\} \times Q$
is a finite set of transitions. A configuration of $\mach$ is a tuple $p(\bu)$ where $\bu \in \{0,1\}^d$.
Given a transition $t = (p,a,q)$ and two configurations $p(\bu)$ and $q(\bv)$, we
say that there is a step from $p(\bu)$ to $q(\bv)$ by means
of the transition $t$, denoted by $p(\bu) \act{t} q(\bv)$ if the following holds:
\begin{itemize}
	\item If $a = \text{test}_j(i)$, then $\bu(i) = j$ and $\bv = \bu$
	\item If $a = \text{set}_j(i)$, then $\bv(i) = j$ and $\bu(k) = \bv(k)$ for every $k \neq i$.
\end{itemize}

The notion of runs, reachability and state-reachability can then be easily defined
for Boolean programs. It is known that the state-reachability problem for Boolean programs
is PSPACE-hard~\cite{tacas/GodefroidY13}. We now give a reduction from the state-reachability problem for Boolean
programs to the state-reachability problem for continuous VASS with resets.

Suppose $\mach = (Q,T)$ is a $d$-Boolean program. We will construct a $2d$-continuous VASS
with resets $\mach' = (Q',T')$ that will simulate $\mach$. Intuitively, for each Boolean variable
$i$ of $\mach$, we will have two counters $i$ and $d+i$ in $\mach'$. Exactly one of these two counters
will be non-zero at any given point. If counter $i$ is non-zero, then this corresponds to the Boolean
variable $i$ being true in $\mach$ and if counter $d+i$ is non-zero, then this corresponds to $i$ being false in $\mach$. Tests and sets on Boolean variables can be accomplished by appropriate increments/decrements
and resets on counters in $\mach'$.

Recall that for any $1 \le i \le 2d$, the matrix $\mathbf{R}^{(2d,i)}$ is the $2d\times 2d$ diagonal matrix
over $\{0,1\}$ which is 1 everywhere on its diagonal, except in the $i^{th}$ position, where it is a 0.
Intuitively, multiplying by $\mathbf{R}^{(2d,i)}$ corresponds to resetting the value
of the $i^{th}$ counter of $\mach'$ and not changing the values of the other counters.
With this in mind, we now describe the states and transitions of $\mach'$.

Formally, the set of states $Q'$ of $\mach'$ is $Q \cup \{q_t : t \in T\}$.
Corresponding to each transition $t$ of $\mach$, we will have two transitions $t^s,t^e$ in $\mach'$
in the following way:

\begin{itemize}
	\item If $t = (p,\text{test}_j(i),p')$ for some $1 \le i \le d$ and some $j \in \{0,1\}$,
	then $t^s = (p,\bI,-\mathbf{e}_{d j + i},q_t)$ and $t^e = (q_t,\bI,\mathbf{e}_{d j + i},p')$. Intuitively, if we want to test that the value of the $i^{th}$ Boolean variable is $j$,
	then we first decrement the counter $d j + i$ by some value and
	then increment it again. Recall that the Boolean variable $i$ having value $j$ in $\mach$
	corresponds to the counter $d j + i$ having a non-zero value in $\mach'$
	and these transitions check exactly that criterion.
	
	\item If $t = (p,\text{set}_j(i),p')$ for some $1 \le i \le d$ and some $j \in \{0,1\}$,
	then $t^s = (p,\mathbf{R}^{(2d,d (1-j) + i)},\mathbf{e}_{d j + i},q_t)$
	and $t^e = (q_t,\bI,\bzero,p')$. Intuitively, setting the value of the Boolean variable $i$ to $j$
	corresponds to first resetting the value of the counter $d (1-j) + i$ to zero
	and then ensuring that the value of the counter $d j + i$ is non-zero.
\end{itemize}

This completes the description of $\mach'$. A configuration $p(\bu)$ of $\mach'$ is called
good if $p \in Q$ and for every $1 \le i \le d$, exactly one of the counters $i$ and $d+i$ have a non-zero value.
Each good configuration $C$ of $\mach'$ is easily seen to correspond to a unique configuration $\mathbb{B}(C)$ of $\mach$.
It can then be easily seen that if $C \act{\alpha t^s, \beta t^e} C'$ is a run in $\mach'$ from
some good configuration $C$ then $C'$ is also a good configuration and 
moreover this run can happen 
if and only if $\mathbb{B}(C) \act{t} \mathbb{B}(C')$ is a step in $\mach$.
This then allows for a straightforward reduction from the state-reachability problem for $\mach$
to the state-reachability problem for $\mach'$.  

\section{Proofs of Subsection~\ref{subsec:lower-bound-perm}}\label{sec:appendix-lower-bound-perm}

Let $\cC$ contain a non-trivial permutation matrix $\bP_\sigma$ (of dimension $k \times k$ for some $k$). 
Since $\bP_\sigma$ is not the identity matrix,
there exists a sequence of distinct indices $i_1, i_2, \dots, i_\ell$ with $\ell \ge 2$ such that $\sigma$ maps $i_1$ to $i_2$, $i_2$ to $i_3$ and so on and finally maps $i_\ell$ to $i_1$.
Let $z = i_1$. Further, let $n \ge 2$ be the least number such that $\bP_\sigma^n = \bI$. Note that since
$\bP_\sigma$ is a permutation matrix, such a number must exist. 

Now, suppose $\bu$ is a vector in $\qnz^k$ such that $\bu(j) = 0$ for every $j \neq z$. Then, 
\begin{equation}\label{eq:perm-one-app}
	\bP_\sigma \cdot \bu = \bv \text{ where } \bv(j) = 0 \text{ for every } j \neq \sigma(z)
\end{equation}

Similarly, suppose $\bu$ is a vector in $\qnz^k$ such that $\bu(j) = 0$ for every $j \neq \sigma(z)$.
Then,
\begin{equation}\label{eq:perm-two-app}
	\bP_\sigma^{n-1} \cdot \bu = \bv \text{ where } \bv(j) = 0 \text{ for every } j \neq z
\end{equation}

Intuitively, $\bP_\sigma$ transfers the value of $z$ to $\sigma(z)$ and does nothing to the other counters, provided the other counters all have the value 0.
Similarly, $\bP_\sigma^{n-1}$ 
transfers the value of $\sigma(z)$ to $z$ and does nothing to the other counters, provided the 
other counters all have the value 0. We will use these equations to prove our desired hardness results.

\subsection{\PSPACE\  lower bound for state-reachability}

First, we will show that the state-reachability problem for $\cC$-continuous VASS is PSPACE-hard
by giving a reduction from the state-reachability problem for Boolean programs.
To this end, let $\mach = (Q,T)$ be a $d$-Boolean program. We will construct
a $dk$-dimensional $\cC$-continuous VASS which will simulate $\mach$. 
The counters $x_1 = z, x_2 = z+k, \dots, x_d = z+(d-1)k$ will be called the positive primary counters of $\mach'$,
the counters $\overline{x_1} = \sigma(z), \overline{x_2} = \sigma(z) + k, \dots, \overline{x_d} = \sigma(z) + (d-1)k$ will be called 
the negative secondary counters of $\mach'$
and the remaining counters will be called dummy counters, which will always have the value zero.
The intuition here is that
the Boolean variable $i$ of $\mach$ will be simulated by the positive primary counter $x_i$ and the negative primary counter $\overline{x_i}$ in $\mach'$. 
At any point, we will have the following invariant in $\mach'$: 
Exactly one of the counters $x_i$ and $\overline{x_i}$ will have a non-zero value.
If $x_i$ (resp. $\overline{x_i}$) has a non-zero value, then this corresponds to the Boolean variable $i$ being true (resp. false).

Note that with this invariant, tests on the Boolean variable $i$ in $\mach$ can be simulated by appropriate increments/decrements on the primary counters $x_i$ and $\overline{x_i}$. 
On the other hand, 
to simulate set transitions on the Boolean variable $i$ in $\mach'$, we will rely on  equations~\ref{eq:perm-one-app} and~\ref{eq:perm-two-app}. 
The idea is that if we want to set the Boolean variable $i$ to 1 in $\mach$, 
then we have to make the value of $x_i$ to be non-zero and the value of $\overline{x_i}$ to be zero. We first check which of these current values is non-zero, by non-deterministically decrementing either $x_i$ or $\overline{x_i}$. Note that exactly
one of these choices can succeed, thanks to our invariant. 
If we successfully decrement $x_i$, then it turns out that the original value of $x_i$ was already non-zero and so we simply have to
increment $x_i$, in order to fully simulate the set transition on the Boolean variable $i$.
On the other hand, suppose we decrement $\overline{x_i}$. Then this means that originally the
value of $\overline{x_i}$ was non-zero and the value of $x_i$ was zero. We then increment
$\overline{x_i}$ and then transfer the value of $\overline{x_i}$ to $x_i$ by multiplying with
$\apply{dk}{\bP_\sigma^{n-1}}{(i-1)k+1}$. By the definition of the Apply operation and
equation~\ref{eq:perm-two-app}, it follows that the value of $x_i$ is now non-zero, the value
of $\overline{x_i}$ is zero and the value of every dummy counter is 0 and so we have successfully simulated a set transition to the value 1.
Similarly by using equation~\ref{eq:perm-one-app}, we can simulate a set transition to the value 0.

Formally, corresponding to each transition $t$ of $\mach$, we will construct transitions in $\mach'$
as follows:
\begin{itemize}
	\item If $t = (p,\text{test}_j(i),p')$, then we construct two transitions 
	$t^s = (p,\bI_{dk},-\mathbf{e}_{x},q_t)$, 
	$t^e = (q_t,\bI_{dk},\mathbf{e}_{x},p')$ 
	where $x = x_i$ if $j = 1$ and $x = \overline{x_i}$ if $j = 0$.
	Note that this construction is very similar to the one given for continuous VASS with resets.
	We test that the value of the Boolean variable $i$ is $j$ by decrementing and incrementing the 
	value of the appropriate primary counter.
	\item If $t = (p,\text{set}_j(i),p')$, then we construct the gadget in Figure~\ref{fig:perm-gadget}.
	Note that in the figure $\bA^+ = \bI_{dk}$ and $\bA^-= \apply{dk}{\bP_\sigma^{n-1}}{(i-1)k+1}$ if $j = 1$
	and $\bA^- = \bI_{dk}$ and $\bA^+ = \apply{dk}{\bP_\sigma}{(i-1)k+1}$ if $j = 0$. The intuition behind
	this gadget was already discussed in the previous paragraphs.
	\begin{figure}
		\begin{center}
			\tikzstyle{node}=[circle,draw=black,thick,minimum size=12mm,inner sep=0.75mm,font=\normalsize]
			\tikzstyle{edgelabelabove}=[sloped, above, align= center]
			\tikzstyle{edgelabelbelow}=[sloped, below, align= center]
			\begin{tikzpicture}[->,node distance = 1cm,scale=0.8, every node/.style={scale=0.8}]
				\node[node, initial, initial text = \text{}] (q0) {$p$};
				\node[node, below right = of q0] (q1) {$q_t^{0-}$};
				\node[node, above right = of q0] (q2) {$q_t^{1-}$};
				\node[node, right = of q1] (q3) {$q_t^{0+}$};
				\node[node, right = of q2] (q4) {$q_t^{1+}$};
				\node[node, above right = of q3] (q5) {$p'$};
				
				\draw(q0) edge[edgelabelabove] node[above]{$(\bI_{dk},-\mathbf{e_{x_i}})$} (q1);
				\draw(q0) edge[edgelabelabove] node[above]{$(\bI_{dk},-\mathbf{e_{\overline{x_i}}})$} (q2);
				\draw(q1) edge[edgelabelabove] node[above]{$(\bI_{dk},\mathbf{e_{x_i}})$} (q3);
				\draw(q2) edge[edgelabelabove] node[above]{$(\bI_{dk},\mathbf{e_{\overline{x_i}}})$} (q4);
				\draw(q3) edge[edgelabelabove] node[above]{$(\bA^+,\bzero)$} (q5);
				\draw(q4) edge[edgelabelabove] node[above]{$(\bA^-,\bzero)$} (q5);
			\end{tikzpicture}
		\end{center}
		\caption{Gadget corresponding to the transition $t = (p,\text{set}_j(i),p')$. $\bA^+ = \bI$ and $\bA^-= \apply{dk}{\bP^{n-1}_\sigma}{(i-1)k+1}$ if $j = 1$
			and $\bA^- = \bI$ and $\bA^+ = \apply{dk}{\bP_\sigma}{(i-1)k+1}$ if $j = 0$.}
		\label{fig:perm-gadget}
	\end{figure}
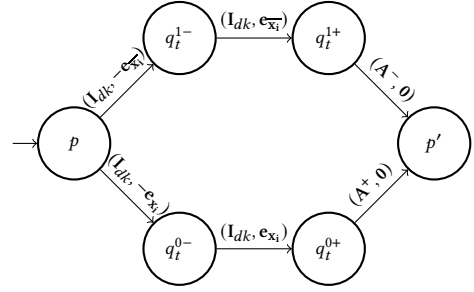 
\end{itemize}

A configuration $p(\bu)$ of $\mach'$ is called good if $p$ is a state of $\mach$ and
for each $i$, exactly one of the counters $x_i$ and $\overline{x_i}$ has a non-zero value
and all the dummy counters have value 0. It is easy to see that every good configuration $C$ of $\mach'$
uniquely corresponds to a configuration $\mathbb{B}(C)$ of $\mach$.
By our construction it can be verified that if $C$ is a good configuration of $\mach'$
and $C \act{\alpha t^s, \beta t^e} C'$ is a run in $\mach'$ where $t$ is a test transition of $\mach$,
then $C'$ is a good configuration and moreover this run can happen if and only if $\mathbb{B}(C) \act{t} \mathbb{B}(C')$ is a step in $\mach$. Further, a very similar claim can also be verified when
$t$ is a set transition, which sets the value of some Boolean variable $i$ to either 0 or 1.
This then allows us to conclude that state-reachability for $\mach$ can be reduced
to state-reachability for $\mach'$, thereby deriving the desired \PSPACE \ lower bound.

\subsection{NEXPTIME lower bound for coverability and reachability}

Now, we move on to proving \NEXP \ lower bounds for the coverability and reachability problems for $\cC$-continuous VASS.
In the previous paragraphs, we have shown that $\cC$-continuous VASS can simulate Boolean programs.
We can enrich the model of Boolean programs by allowing continuous counters to get the model
of continuous Boolean programs where, in addition to a set of Boolean variables we are also
given access to a set of continuous counters. From~\cite{pacmpl/BalasubramanianMTZ24}, it is known that the 
coverability and reachability problems for continuous Boolean programs are \NEXP-hard. 
The reduction given in the previous subsection then allows us to derive a \NEXP \ lower bound for the 
coverability and reachability problems for $\cC$-continuous VASS in a straightforward manner.


\end{document}